\documentclass{article}
\usepackage{kr}
\usepackage{times}
\usepackage{soul}
\usepackage{url}
\usepackage[hidelinks]{hyperref}
\usepackage[utf8]{inputenc}
\usepackage[small]{caption}
\usepackage{graphicx}
\usepackage{amsmath,dsfont}
\usepackage{stmaryrd} % for an arrow
\usepackage{booktabs}
\usepackage{algorithm}
\usepackage{algorithmic}
\usepackage{mathabx}
\usepackage{amsmath,verbatim,bussproofs}

\newcommand{\eg}{\emph{e.g.}, }
\newcommand{\ie}{\emph{i.e.}, }
\newcommand{\rr}{{\mathit rn}}
\newcommand{\dob}{\bigcirc}
\newcommand{\kb}{\mathbf{R}}
\newcommand{\kbc}{\mathbf{R}^{\Rightarrow}}
\newcommand{\kbo}{\mathbf{R}^{\mathrm{o}}}

\usepackage{gn-logic14}
\DeclareSymbolFont{symbolsC}{U}{txsyc}{m}{n}
\DeclareMathSymbol{\strictif}{\mathrel}{symbolsC}{74}
\DeclareMathSymbol{\boxright}{\mathrel}{symbolsC}{128}
\DeclareMathSymbol{\boxRight}{\mathrel}{symbolsC}{136} % Lewis’s stronger ‘would’ counterfactual
\DeclareMathSymbol{\diamondRight}{\mathrel}{symbolsC}{140} % Lewis’s stronger ‘might’ counterfactual
\DeclareMathSymbol{\diamonddot}{\mathord}{symbolsC}{144} % Lewis’s inner necessity

\newcommand{\m}{\mathrm{m}}
\newcommand{\PL}{\mathrm{PL}}

\newcommand{\maxf}{\mbox{\it maxf}}

\newcommand{\outf}{\mbox{\it outf}}
\newcommand{\modpl}{\models_{\rm PL}}

\usepackage{latexsym}
\usepackage{hyperref}
\usepackage{thmtools,amsthm} %restatable
\usepackage{thm-restate}
\newtheorem{example}{Example}
\newtheorem{theorem}{Theorem}
\newtheorem{definition}{Definition}
\newtheorem{fact}{Fact}
\newtheorem{assumption}{Assumption}
\newtheorem{remark}{Remark}
\newtheorem{proposition}{Proposition}
\newtheorem{lemma}{Lemma}
\usepackage{cancel}
\title{Defeasible Conditional Obligation in  a Two-tiered Preference-based Semantics (Extended Version)}

\author{% 
Xavier Parent
    \affiliations
    Institute of Logic and Computation, TU Wien, Vienna, Austria
    \emails
    xavier.parent@tuwien.ac.at
}
\begin{document}

\maketitle

\begin{abstract}

%This paper addresses a problem posed by Horty by providing a preference-based, two-tiered account of defeasible conditional obligation. The two-tiered structure is reflected in two respects. First, on the semantic level,
%a nonmonotonic reasoning mechanism is added to a monotonic logic, using a set of conditionals representing the normative system. 
%%a nonmonotonic reasoning mechanism is obtained by parameterizing an otherwise monotonic consequence relation with a set of conditionals  representing the normative system. 
%Second, ideality and normality orderings in the models to avoid problems in earlier approaches, advocating a distinct ranking method for each.
%%This paper introduces an ideality and normality ordering to avoid problems in earlier approaches.

%In response to a problem posed by Horty, this paper presents a two-tiered account of %defeasible conditional obligation within a preference-based framework. 

In response to a concern raised by Horty, this paper 
develops a two-tiered, preference-based semantic framework for modeling defeasible conditional obligations. 
The paper extends a Hansson–Lewis style preference semantics for dyadic deontic logic by incorporating a nonmonotonic reasoning mechanism that enables previously derived obligations to be withdrawn when new, potentially conflicting information comes in. 
%provides a two-tiered, preference-based semantic account of defeasible conditional obligation. It extends a dyadic deontic logic with a nonmonotonic reasoning mechanism that allows for the withdrawal of obligations in  light of new information.
%In response to a problem posed by Horty, this paper proposes a two-%tiered, preference-based semantic account of defeasible conditional %obligation.
%It extends a dyadic deontic logic with a nonmonotonic reasoning %mechanism,
%which permits the retraction of obligations in the light of new information.
The account is bi-preferential: two orderings$-$ideality and normality$-$on worlds are employed to address shortcomings in earlier approaches, with a separate ranking method for each. 
At the nonmonotonic layer, a number of postulates are considered, including antecedent strengthening, inclusion and no-drowning. A connection is established with so-called constrained input/output (I/O) logic$-$an existing standard for normative reasoning based on a different methodology.
\end{abstract}

\section{Introduction}

Obligations are inherently conditional and take the form of conditional statements. Their logical analysis is the focus of the logics of conditional obligation or dyadic deontic logics. These can be classified into two categories: preference-based systems and norm-based systems.
Preference-based systems analyse the conditional obligation operator within the framework of possible world
semantics.  A  preference relation ranks possible worlds in terms of betterness or ideality. These systems can be seen as the deontic cousin of the KLM systems for nonmonotonic inference  \cite{ddl:KLM90}.\footnote{\cite{Makinson93} describes them as a nonmonotonic formalism that was ``ahead of its time'' (\emph{avant la lettre}) and identifies them as the fifth—and ``most complex'' (p. 374)—face of minimality.} 
%They can be considered as the deontic analog of the KLM systems for nonmonotonic inference operations.  Makinson \cite{Makinson93} describes  them as a nonmonotonic formalism ahead of its
%time (``avant la lettre'').
 Notable systems in this category include: \cite{H69}'s DSDL systems; \cite{ddl:vanFraassen1972}'s system CD; \cite{L73}'s system VN; \cite{Gob03}'s system DP;  \cite{ddl:A02}'s systems E, F and G; and \cite{ddl:kr12}'s semantics.  In contrast, norm-based systems analyse the conditional obligation operator with reference to a set of explicit norms representing the normative system. The semantics is operational rather than model-theoretic. The meaning  of the conditional obligation operator is given through a  set of procedures that yield outputs based on given inputs. Prominent systems in this category are: \cite{horty12}'s  deontic default logic; \cite{MakinsonT00}'s input/output (I/O) logic; and \cite{Hansen08}'s and \cite{serg25}'s imperativist semantics for deontic logic.

The methodologies in these two groups differ substantially, leading to distinct semantic frameworks. Each type excels in different areas.  Preference-based systems are particularly effective at handling contrary-to-duty (CTD) reasoning, such as \cite{chisholm}’s paradox--a core problem in deontic logic dealing with norm violation. Norm-based systems are well-suited for defeasible, or nonmonotonic, reasoning: they can accommodate norms with exceptions, allowing conclusions to be revised or withdrawn as new information becomes available.
Unifying these approaches so as to combine their respective strengths remains a central challenge,  identified as an open %research question 
problem by \cite{Horty2014}.

This paper focuses on concerns raised by Horty
\cite{Horty1993,Horty1994,Horty07,horty12,Horty2014} concerning the ability of preference-based systems to accurately model defeasible normative reasoning. 
The objections are summarized, and a new ordering semantics is proposed to address them. The approach is two-tiered in two respects:
\begin{itemize}
\item[(i)]  A preference-based dyadic deontic logic is extended with a nonmonotonic reasoning mechanism; nonmonotonicity is obtained by (indirectly) parameterizing an otherwise monotonic consequence relation with a set of norms representing the normative system;
\item[(ii)] Two orderings--ideality and normality--on worlds are employed to address issues in previous implementations, notably the fallacy of the prohibited exception (see \S \ref{fal:exc}); each ordering uses a distinct ranking method.
%The paper introduces two orderings—one representing
%  ideality, the other normality—in the models to resolve problems  in %earlier implementations of (i), such as those of %\cite{Lehmann95,Delgrande20}, most notably the fallacy of the %prohibited exception (see \S \ref{fal:exc}). A different ranking method %for each type of ordering is advocated.
%namely the mistaken treatment of an exception to a norm as itself prohibited. (see \S %\ref{fal:exc}).
\end{itemize}
%The approach is two-tiered in two respects.
%\begin{itemize}
%\item [(i)] A nonmonotonic reasoning mechanism is obtained by parameterizing an otherwise monotonic consequence relation with a set of conditionals representing the normative system;
%\item [(ii)] Models are equipped with two orderings, capturing ideality and normality. This is needed to avoid problems in previous attempts for (i). 
%%A nonmonotonic mechanism is imposed on a monotonic framework by (indirectly) parameterizing the consequence relation with a set of conditionals, representing, for example, the underlying normative system.
%\end{itemize}
The deontic fragment of the proposed framework is  faithfully embedded into constrained input/output (I/O) logic,
a well-established norm-based system due to \cite{DBLP:journals/jphil/MakinsonT01}.\footnote{Overview in \cite{Parent2013b}.} 
Earlier approaches either focus on so-called unconstrained I/O logic, which does not support nonmonotonic nor CTD reasoning \cite{MakinsonT00,Boch05,ParentT14,Stolpe20,Parent21,CiabattoniR23}, or adopt a proof-theoretic perspective based on adaptive logic or argumentation \cite{StrasserBP16,Arieli0S24}.
This paper adopts a model-theoretic, KLM-style perspective on constrained I/O logic, extending previous findings for the unconstrained case. To my knowledge, it is the first endeavor of its kind.\footnote{This paper is an extended version of \cite{Parent26b}.}

\section{Preliminaries} \label{h:op}

%that combine the advantages
%of both. 
Horty's %skepticism touches on some intricate aspects of logic and semantics. The 
main point is that normative reasoning is defeasible, and thus nonmonotonic: a conclusion drawn can be revised by new information. %$-$a type of non-monotonic reasoning where conclusions can be invalidated by further evidence. 
Preference-based dyadic deontic logic was developed as a species of modal and conditional logic.  The entailment (or consequence) relation is monotonic: once conclusions are drawn, they cannot be revoked. 
%This makes this logic unsuitable for normative reasoning. 

Nonmonotonicity may be understood as a property of either the conditional, in the object language, or the entailment relation, in the meta-language. Following Horty and others, this paper adopts the latter perspective, justifying it with reference to what may be called the strengthening problem, a central issue in normative reasoning

%The property of nonmonotonicity can, at first glance, be located either in the object language, as a property of the conditional, or in the meta-language, as a property of the entailment relation. Following Horty and others, I adopt the second approach and justify this choice by appealing to what I will call the strengthening problem, which is central to normative reasoning.

%Following Horty and others, I adopt the second approach and justify this choice as follows, with reference to (as I will call it) the strengthening problem. I consider it as a central problem for normative reasoning.

%The property of nonmonotonicity can be either placed at the object-language (as a %property of the conditional) or in the meta-language (as a property of the %entailment relation). Following Horty and others, I go the second, and justify %this choice as follows.

%This section has three parts. First, it lays out Horty’s problem in detail. Next, %it looks at one possible response$-$the enthymematic approach$-$and explains why %it is not adopted. Finally, it gives an independent argument for using different %ranking methods for deontic and normality conditionals, which is an important part %of the solution being proposed.

\subsection{The Strengthening Problem} \label{fal:exc}
%Horty's work \cite{Horty1993,Horty1994,Horty07,Horty2014}, in particular some of his objections in \cite{Horty2014} to the classical semantics of deontic modals, as exemplified in %Kratzer \cite{Krat12}'s work. Horty's skepticism touches on some intricate aspects of logic and semantics. The main point is that defeasible reasoning can be revised or overridden by %new information$-$a type of non-monotonic reasoning where conclusions can be invalidated by further evidence. In contrast, the consequence relation of preference-based semantics for %dyadic deontic logic is monotonic: once conclusions are drawn, they cannot be revoked.
  
%It is time to clear up a potential misunderstanding. 

%Nonmonotonicity can be either placed at the object-language, as a property of the 

$\bigcirc (B/A)$ may be read as ``It ought to be the case that $B$ if $A$". The principle of \emph{Strengthening of the Antecedent}~(SA) allows one to infer $\bigcirc (B/A\wedge C)$ from $\bigcirc (B/A)$. 
This principle is central to normative reasoning: $\bigcirc (B/A)$ is general to the extent that $A$ can be ``freely'' strengthened. % and made more specific. 
   In practice, $A$ occurs with some background condition 
$C$. The question is whether $B$
remains obligatory. By default, it does$-$unless 
$C$ cancels the obligation.
Without the ability to strengthen antecedents, obligations would fail as cues for action. %For instance, a child told ``Do not eat with your fingers!" naturally assumes that the %prohibition  still holds in more specific contexts, i.e., ``if eating asparagus.''
%,’ for example—because he freely strengthens the antecedent. 

The strengthening problem concerns finding a middle ground between unrestricted strengthening and none
at all.
%balancing unrestricted strengthening of antecedents with no strengthening at all. 
 A child  told ``Do not eat with your fingers'' naturally assumes that the prohibition extends to more specific contexts, like ``if eating asparagus,'' unless  instructed otherwise. Accordingly, antecedents can be strengthened as much as possible, 
treating the obligation as generally as possible—two ways of expressing the same idea—unless doing so conflicts with available evidence.
%, which prevents deriving conclusions contradicted by known facts.
%which ensures the framework avoids drawing conclusions contradicted by known facts.
%For example, a child told ``Do not eat with your fingers'' naturally assumes the prohibition still holds in more specific contexts, such as ``if eating asparagus,'' unless explicitly instructed otherwise.
%unless overridden by evidence to the contrary. 
%The strengthening problem is the problem of finding the ``right'' balance between unrestricted strengthening and no strengthening at all. A natural approach is to permit the strengthening as much as possible and treat the obligation \emph{as generally as possible}, unless explicitly instructed otherwise. 
%The solution advocated by horty consists in  allowing it ``as much as possible", unless it generates trouble, particularly if there is evidence to the contrary. 
Implementing this idea requires a nonmonotonic entailment relation, which existing norm-based systems can handle easily. The question is whether preference-based dyadic deontic logic 
can accommodate the same kind of defeasible reasoning.

%applies when served asparagus; he ``freely'' strengthens the antecedent. It is only when a more specific overriding obligation is added  (``Eat with your fingers when eating asparagus'') that the default %inference is blocked.
%The solution to the strengthening problem developed in this paper is in line with established ideas in KRR.  It consists in allowing the strengthening ``as much as possible", unless it generates trouble, %particularly if there is evidence to the contrary. This necessitates a nonmonotonic entailment relation.  

The proposed  approach differs from the ``enthymematic'' treatment of the strengthening problem \cite{ddl:G04,BGL14}. The latter keeps a monotonic entailment relation, but restricts (SA) suitably.  A seemingly valid strengthening is accounted for by making a ``hidden'' premise$-$typically, a permission--explicit. Due to space limitation, a discussion of the enthymematic approach must be deferred to another occasion.

%%This example shows that the entailment relation must be nonmonotonic: the addition of new information may yield to the withdrawal of a previously drawn conclusion.   %I recall it briefly.
To illustrate the strengthening problem, Horty uses the following example.\footnote{As usual $\bigcirc A$ is short for  $\bigcirc (A/\top)$, where $\top$ is a tautology.}
%argue that a certain amount of strenghening is needed; the example is handled similarly using the previous defeasible form of antecedent strenghening. 

\begin{example}[The asparagus, \cite{Horty1993}] \label{asparagus} Consider the following rules of etiquette:
\begin{description}%\vspace{0.2cm}
\item [{\rm (i)}]  You ought not to eat with your fingers: $\bigcirc \neg \mathit{f}$;
\item [{\rm (ii)}]  You ought to put your napkin on your lap: $\bigcirc\mathit{n}$;
\item [{\rm (iii)}]  If you are served asparagus, you ought to eat with your fingers: $\bigcirc(\mathit{f}/\mathit{a})$.
\end{description}%\vspace{0.2cm}
Horty notes that, assuming (i) exhausts the premise set, (iv) should be derivable: 
%one should be able to derive 
\begin{description}%\vspace{0.2cm}
\item [{\rm (iv)}]  If you are served asparagus, you ought not to eat with your fingers: $\bigcirc (\neg\mathit{f}/\mathit{a})$.
\end{description}%\vspace{0.2cm}
He goes on to observe that, if the premise set is (i)-(iii), then (iv) should not be derivable$-$(iv) contradicts (iii). Accordingly, 
(SA) cannot be endorsed as it is.  However, some strengthening is necessary: just as one wants to infer (iv) from (i) alone, one  also wants to infer (v) from (i)–(iii).
\begin{description}%\vspace{0.2cm}
\item [{\rm (v)}]  If you are served asparagus, you ought  to put your napkin on your lap: $\bigcirc (n/\mathit{a})$.
\end{description}%\vspace{0.2cm}
The reason is that (iii) defeats or overrides (i), but not (ii). Hence, a conditionalization on
$a$ should be permitted for (ii), but not for (i).
%The intuition is that (iii) overrides (i), but not (ii), so that one should be able to %conditionalize
%(ii) to $\mathit{a}$, but not (i).
 %$A > (\cdot)$ 
\end{example}

Techniques were developed in Knowledge Representation and Reasoning (KR) to resolve a similar problem.\footnote{\label{myfootnote}
The asparagus scenario resembles what Pearl calls the problem of “property inheritance [blocking] from classes to subclasses” \cite[p. 128]{Pearl90}. If a subclass of $C$ (e.g., penguins as birds) is exceptional with respect to one property (flying), it may fail to inherit other typical properties of $C$ (such as having wings). This illustrates what \cite{Lehmann95} calls ``presumptive reasoning'': properties of a class are assumed to hold for its members unless there is evidence to the contrary. Intuitively, one wants to infer ``penguins have wings'' from ``birds have wings''. The formalization in \cite[Ex.~4]{CasiniMV19} further clarifies the analogy between the two cases.}
%The asparagus scenario resembles what Pearl calls the problem of “property inheritance [blocking] from classes to subclasses” \cite[p. 128]{Pearl90}. If a subclass of $C$ (e.g., penguins as birds) is exceptional with respect to one property (flying), it may fail to inherit other typical properties of $C$ (such as having wings).
%This illustrates what  \cite{Lehmann95} calls ``presumptive reasoning'': properties of a class are assumed to hold for its members unless there is evidence to the contrary. Ideally, one would like to infer from ``birds have wings'' that ``penguins have wings.'' The formalisation in \cite[Ex. 4]{CasiniMV19} makes the analogy between the two examples even more striking. } An example is \cite{Lehmann95}'s lexicographic entailment relation. The key idea is to  constrain the ideality relation in the models further using a set of designated conditionals.\footnote{A similar proposal was made in dyadic deontic logic by \cite{ddl:PS97}, but for a different purpose. 
 %A more thorough comparison between the account described in this paper and theirs is left as a topic for future research.
% } 
  %To cope with the later problem, he basic idea is to refine the ordering in a typical model of the premise set. This gives more control over the obligation sentences outputted as solutions. Other similar techniques have been developed in nonmonotonic logic. I'll focus on a variant recently put forth by Delgrande \cite{Delgrande20}. 
While one might consider applying them to deontic logic, Horty dismisses this idea: %dismisses this idea:
%One could think that they could be used in a deontic setting. However these solutions are dismissed by Horty:

%
%\begin{center}
%	\begin{minipage}{12cm}
\begin{quote}
``Although I can offer no argument to show that this cannot be done, I have not seen it done successfully, and I am skeptical: being served cold asparagus does not seem like the sort of thing that should force us to consider nonideal worlds.'' \cite[p. 449]{Horty2014}
\end{quote}
%\end{minipage}
%\end{center}

%\noindent
Horty's reference to ``being served cold asparagus'' highlights the following point:
%The sentence ``being served cold asparagus does not seem like the sort of thing that should force us to consider non-ideal worlds'' refers to the following fact:
% that in Kratzer's semantics (and in Hansson's semantics), $\{\bigcirc \neg f, \bigcirc (f/a)\}$ entails $\bigcirc \neg a$, and hence $a$ can only be true in nonideal worlds. 

\begin{quote}
``In the classical [viz., preference-based] framework one cannot have $\bigcirc \neg f$ and $\bigcirc (f/a)$ without  $\bigcirc \neg a$, because the first two statements entail the third." \cite[p. 446]{Horty2014}
\end{quote}
 The truth-conditions put $\bigcirc (B/A)$  true if all the best $A$-worlds are $B$-worlds. Accordingly, $\bigcirc \neg f$ states that $f$ is false in all the best $\top$-worlds, and $\bigcirc (f/a)$ states that  all the best $a$-worlds are $f$-worlds. It follows that $a$ is false in all the best $\top$-worlds; that is, $\bigcirc \neg a$.\footnote{Suppose not. Then
 a best $\top$-world $w$ is an $a$-world. Since $\bigcirc \neg f$,  $w$ makes $f$ false.  However, by $\bigcirc (f/a)$, no  best $a$-world makes $f$ false, so $w$ is not a best $a$-world.  It follows that $w$ is not a best $\top$-world, in contradiction with our assumption.  The argument  is informal, and can be verified more rigorously using the evaluation rule
 that identifies the best worlds with maximal ones (cf. \S \ref{account}).} Hence, $a$-worlds are ``non-ideal" overall. The exception (``being served asparagus'') is treated as a prohibition for the only reason that it is an exception.

% \vspace{-0.2cm}
%\begin{center}
%	\begin{figure}[h]
%		\includegraphics{pic/f-dis}\caption{The prohibited exception
%		%f-DIS ranking.  A bar ``$\bar{\;}$'' over a propositional letter means ``$\neg$''. Falsifying a more specific conditional does not reduce a world's normality.
%		%The exceptional worlds each falsify one conditional. 	
%		}\label{tab0}
%	\end{figure}
%\end{center}
%\vspace{-0.5cm}
%
%\bigskip
%\begin{minipage}{0.5\textwidth}
%\centering
%\begin{tabular}{ll}
%best: & $ \bar{f}$ \\
%2nd-best: & $f,a$% $\qquad $ $r, \bar{a}$   \qquad  $\bar{r}, a$
%\end{tabular}
%%\captionof{table}{f-DIS ranking.  The exceptional worlds each falsify one conditional. A bar ``$\bar{\;}$'' over a propositional letter means ``$\neg$''.}\label{tab1}
%\end{minipage}
%%\%end{center}
%\bigskip

%and hence its negation must be true there.  $a$ cannot be true in any of the best possible worlds, because $\bigcirc \neg f$ tells is that $\neg f$ is true there, and $\bigcirc (f/a)$ tells us that  the best $a$-worlds are $f$-worlds. 

% its negation holds inIf this were not the case, %$\bigcirc (f/a)$. 
%The classical semantics transforms the exception into a prohibition: being served asparagus ($a$, the exception) is forbidden, since its negation holds in all the best worlds. If this was not so, the best $a$ would not be $f$. 

In \cite{Ryle38}'s terminology, such an outcome constitutes a category mistake—akin to asserting that ``The number two is blue''—and may be termed the ``fallacy of the prohibited exception.''
 Given its centrality, particularly in \cite{Horty2014}, 
it is presented below as a separate observation.
%I will present it below as aobservation.
This problem affects the classical semantics, but also some of the more recent works in nonmonotonic logic mentioned above, like \cite{Lehmann95}'s lexicographic entailment relation.
%Both are already disqualified by the fact that their consequence relation is non-monotonic. 
%and a number of approaches developed in non-monotonic logic. 
However, it does not apply to \cite{Delgrande20}'s account. %As will be argued in Section \ref{excepviol}, 
But the cost  is significant.
% $\bigcirc \neg a$ is eliminated, but only because $\bigcirc \neg f$ disappears too. 
%The account cannot distinguish a violation from an exception.
The obligation $\bigcirc \neg a$ is lost only because the more general obligation $\bigcirc \neg f$ is also eliminated; in context $\top$, the general obligation simply disappears.
This is the drowning effect \cite{Broersen2012,ParentT17,ParentT18}.\footnote{
This is distinct from the ``drowning" problem described in \cite{ben93}, which is another term for Pearl's problem of property inheritance.}
%
%The drowning effect (for deontics) is discussed in \cite{PS96, ddl:PS97}, although they do not use this term. It is described as a problem for solutions to the contrary-to-duty (CTD) problem based on %nonmonotonic logics (\cite{RL94, McC94}). Prakken and Sergot write: ``in circumstances where a primary obligation is
%violated, consistency is maintained by regarding the derivation of the primary
%obligation as somehow blocked [...] What [this solution] fails to capture is that
%when the secondary obligation, say to kill gently, is being fulfilled, at the
%same time the primary obligation not to kill is being violated: violating an
%obligation in a situation does not make it inapplicable to that situation once violated." \cite[p. 227]{ddl:PS97}. The main difference is that the primary obligation ``disappears" when violated.}
This poses a problem for norm-compliance checking. In certain applications, it may be desirable to determine whether $\bigcirc (B/A)$
has been violated or fulfilled. Intuitively, such a check assumes that the obligation holds, but in the model, the obligation does not hold.\footnote{The drowning effect (for deontics) was originally discussed in \cite[p. 98ff]{PS96}, although they do not use this term. There it is described as a problem for solutions to the contrary-to-duty (CTD) problem based on nonmonotonic logics \cite{RL94,McC94}. In these approaches,
an obligation disappears when violated. 
For example, in the so-called cottage regulation example, if we can no longer infer that there should be no fence (because a fence exists), how do we formally represent that a violation has occurred?
%One may argue that if an obligation is to be preserved after violation, then, a fortiori, it ought to be preserved in a non-violation context as well.
%In a violation context, the violated obligation (e.g., not to kill) disappears. If one %wants it to remain in a violation context, one a fostiori wants it to remain when not violated. 
% if it is desirable it rem
 % They write: ``A crucial effect of letting the primary
% obligation [e.g., not to kill]  be defeated by the secondary obligation [e.g., to kill %gently] is that the primary one
%cannot be violated, since it is simply not applicable to the situation.'' \cite[p. 98]{PS96}
 }

\begin{example}[Prohibited exception, drowning effect]\label{prohib-excep}  From
\begin{description}%\vspace{0.2cm}
\item  [{\rm (i)}] You ought not to eat with your fingers: $\bigcirc \neg \mathit{f}$;
\item  [{\rm (ii)}] If you are served asparagus, you ought to eat with your fingers: $\bigcirc(\mathit{f}/\mathit{a})$;
\end{description}%\vspace{0.2cm}
one wants to be able to infer
\begin{description}%\vspace{0.2cm}
\item  [{\rm (i)}] You ought not to eat with your fingers: $\bigcirc \neg \mathit{f}$;
\end{description} %\vspace{0.2cm}
but not
\begin{description}%\vspace{0.2cm}
\item [{\rm (iv)}]   You ought not to be served asparagus: $\bigcirc\neg \mathit{a}$.
\end{description} %\vspace{0.2cm}
\end{example}

%

%This objection may be called the objection of the "Prohibition of %Exceptionalism". 
%It applies to Lehman's lexicographic entailment, but not to Delgrande's %variant account. However, the price to be paid is high:  $\bigcirc \neg %a$ goes away, but only because $\bigcirc \neg f$ goes away too. 

%The consequence relation is parametrised by two components: deontic and normality %conditionals. 

%The solution to the strengthening problem developed in this paper is in line with established ideas from the literature on default reasoning in logic and computer science. It consists in allowing the %strengthening ``as much as possible", unless it generates trouble, particularly if there is evidence to the contrary. This necessitates a nonmonotonic entailment relation. What's new, however, is its %implementation within a possible worlds semantic framework. This approach differs from the ``enthymematic'' treatment of the problem, which consists in restricting (SA) while keeping a monotonic %entailment relation. Seemingly valid strengthening are then accounted for by making a ``hidden'' premise explicit. Section \ref{enth} will discuss this method in more detail and explain why a different path %has been chosen.

A bi-ordering semantics is employed to address both problems, separating ideality and normality rankings.
%\footnote{This suggestion was previously made in \cite[\S 8]{CPvTT18}, though no details were provided.} 
The language  of dyadic deontic logic  will thus be supplemented with a normality conditional  $\Rightarrow$, expressing what typically holds, relative to a notion of normal or non-exceptional situation.\footnote{A classic reference on the modal approach to normality  is \cite{BOUTILIER199487}.} For $A\Rightarrow B$, read ``If $A$, then normally~$B$". The basic idea is to constrain  the orderings on worlds in a model further using a designated set of $\bigcirc$- and $\Rightarrow$-formulas. While the use of designated obligations to constrain an ordering on worlds is not new, our approach introduces a second ordering together with a separate set of (normality) conditionals associated with it.
%Different ranking methods will be used for normality and ideality. This will be motivated in Section \ref{rank}. 

\subsection{Ranking Worlds} \label{rank}

%For normality,  I will use \cite{Lehmann95}'s lexicographic ordering, while for betterness I will use \cite{Delgrande20}'s method.  I will call them LEX and DEL, respectively.

 %the basic idea is to constrain the ordering in the models to filter out conclusions, or enforce them in a principle way. 
%I will use different ranking methods for normality and betterness. 

 This subsection offers an independent justification for employing distinct ranking methods for deontic and normality conditionals, a key component of the solution to Horty's problem.
 % The two methods revolve around the notion of violation$-$which has become a first-class citizen in deontic logic since the seminal work of \cite{K57,And58}$-$and its counterpart in the modal logic of %normality conditionals. 
 In what follows, the generic term `rule' and the symbol $r$ are used to refer indistinctly to a conditional obligation and a normality conditional.

For normality,  \cite{Lehmann95}'s lexicographic ordering method will be used, while for ideality \cite{Delgrande20}'s method will be used.  They will be called LEX and DEL, respectively. Both are variations on a method of ranking that is familiar from conditional logic and deontic logic.
Assessing the   normality (resp. goodness) of a world involves examining how well it conforms to the normality conditionals (resp. obligations). Fewer violations correspond to a higher degree of normality (resp. ideality).
% It consists in evaluating the normality (resp. goodness) of a possible world  based on the number of conditionals (resp. obligations) this world falsifies (resp. violates). Specifically, a world 
%$w_1$ is considered at least as normal (resp. good) as another world $w_2$ if  it  falsifies (resp. violates) fewer or an equal number of conditionals (resp. obligations) compared to 
% $w_2$.  
This ranking method is called f-DIS, for flat distance-based. The relation captures an abstract notion of distance to ``full'' normality (resp. ideality). The adjective ``flat'' is meant to emphasize the absence of hierarchy among conditionals (resp. obligations).\footnote{Three accounts based on f-DIS are: the semantics of the inference relation of a Poole system \cite[\S 3.1]{FREUND1998209},  \cite{ddl:kr12}'s semantics  and \cite{BGL14}'s theory of P-graphs (for the special case of a ``flat'' P-graph). 
\cite{mak99} discusses how f-DIS relates to \cite{Jorgensen1937}'s dilemma$-$the problem of whether norms bear a truth-value and whether a logic of norms can exist.}  
%f-DIS aligns with the broader possible worlds semantics used in analyzing counterfactuals, as developed by philosophers like \cite{Stalnaker1968-STAATO-5}, \cite{L73}, and others. In their theories, the truth %of a counterfactual depends on the ``closeness'' or similarity of possible worlds to the actual world, with preference given to worlds that minimally diverge from the actual world.  For obligations, what is %measured  is not  the closeness of possible worlds to the actual world, but their closeness to ``full'' ideality, with preference given to worlds that minimally diverge from a world where all the obligations are %fulfilled.  For  normality conditionals, what is measured  is the distance of possible worlds to a world that is as normal as possible. 
%Thus, assessing the  goodness or normality of a world involves examining how well it conforms to the established obligations or conditionals$-$fewer violations correspond to a higher degree of ideality or %normality. This comparison may be made by counting violations (i.e., by cardinality) or via set-theoretic inclusion, or, equivalently, by comparing the set of satisfied conditionals or fulfilled obligations rather %than the falsified or violated ones.\footnote{Like in \cite{ddl:kr12}'s semantics or in \cite{BGL14}'s theory of P-graphs (in the particular case of a ``flat'' P-graph, see Section \ref{p-graph} below).}
% Regardless of the chosen method, 

 A problem arises with rules allowing for exceptions.
 %which are
%commonly formalized by a more specific rule %overriding a more general one.
% an exception invokes a specificity relation: when one antecedent is strictly more specific than another, the specificity rule dictates that the conditional with the
%a rule with a more specific antecedent overrides one with a less specific antecedent. 
Consider the pair of conditionals ``you normally do not eat asparagus'' ($\top \Rightarrow\neg a$) and ``if you are at the restaurant, you normally eat asparagus'' ($r \Rightarrow a$).   A normality conditional (resp. an obligation) $r$ is said to be falsified (resp. violated)  in a world if its antecedent is true and 
its consequent false in this world. In this example, f-DIS yields that a $r\wedge a$-world and  a $r\wedge\neg a$-world are equally exceptional. This is shown in Fig.~\ref{tab1}. Paradoxically, falsifying a more specific conditional does not make a world less normal compared to falsifying a more general one--the two are treated equally.
%reduce a world’s normality. It  
%does not make a world less normal compared to falsifying a more general conditional%—the two are treated equally.
  \vspace{-0.2cm}
\begin{center}
	\begin{figure}[h]
		\includegraphics{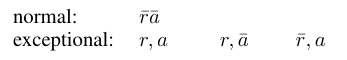}\caption{f-DIS ranking.  A bar ``$\bar{\;}$'' over a propositional letter means ``$\neg$''. Falsifying a more specific conditional does not reduce a world's normality.
		%The exceptional worlds each falsify one conditional. 	
		}\label{tab1}
	\end{figure}
\end{center}
\vspace{-0.5cm}
This outcome is counter-intuitive because, while specificity may represent an exception to a general rule, it still encodes a meaningful expectation.\footnote{The point is even clearer in Tweety's well-known example: ``birds normally fly'' ($b\Rightarrow f$) and ``penguins (like Tweety) normally do not fly''  ($p\wedge b\Rightarrow \neg f$). Although penguins are exceptions to the general flying rule for birds, the more specific conditional ``penguins do not fly'' still conveys an expectation. Falsifying  it (even artificially, with the help of artificial wings) results in an even less normal world.}  It is noteworthy that, under the  standard truth-conditions for $\Rightarrow$ in terms of maximality (viz. ``all the most normal $A$-worlds are $B$-worlds"), the resulting model makes $r \Rightarrow a$ false. Indeed, a most normal $r$-world is  a $\neg a$-world. 

%\bigskip
%\begin{minipage}{0.5\textwidth}
%\centering
%\begin{tabular}{ll}
%normal: & $ \bar{r}\bar{a}$ \\
%exceptional: & $r,a$ $\qquad $ $r, \bar{a}$   \qquad  $\bar{r}, a$
%\end{tabular}\captionof{table}{f-DIS ranking.  The exceptional worlds each falsify one conditional. A bar ``$\bar{\;}$'' over a propositional letter means ``$\neg$''.}\label{tab1}
%\end{minipage}
%%\%end{center}
%\bigskip

%\begin{tabular}{ll}
%normal: & $\neg a \neg f$ \\
%exception: & $a,f$ $\mid $ $a, \neg f$   $\mid $ $\neg a, f$
%\end{tabular}

LEX addresses this problem by ranking worlds lexicographically rather than by a simple count. A world is more normal if it falsifies fewer conditionals, giving precedence to more specific conditionals; falsifications of less specific conditionals are considered only in case of a tie on more specific conditionals. 
%worlds that satisfy more specific conditionals are considered preferable to worlds that %satisfy a greater number of less specific conditionals.
%corrects this problem as follows. Instead of a simple count, one proceeds lexicographically: satisfying as many more specific conditionals as possible is prioritized over satisfying a larger number of less specific conditionals. 
One assumes a partition of the set of (explicitly given) conditionals into levels of specificity, and applies f-DIS level by level, starting with the set of most specific conditionals and proceeding up to the less specific ones.  The ranking is shown in Fig.~\ref{table2} in Ex.~\ref{ex:norm}.

LEX is appropriate for normality conditionals, but not for deontic conditionals. Consider: ``you should not eat with your fingers" ($\bigcirc\neg f$) and ``if you are served asparagus, you should eat with your fingers" ($\bigcirc (f/a)$). f-DIS and LEX yield that the $\neg a\wedge\neg f$-world  is strictly better than the $a\wedge f$-world. This clashes with the specificity principle, which states that the more specific obligation overrides and ``suspends'' the more general one. $\bigcirc (f/a)$ states that, if $a$ is the case, then $f$ is obligatory. Intuitively, $\bigcirc\neg f$ no longer applies. Why, then, consider this obligation violated, if $a\wedge f$? This is counter-intuitive.
%, and generates the fallacy of the prohibited exception. 
DEL corrects this problem by stipulating that $\bigcirc (B/A)$ is not considered violated in world $w$—even when $A$ holds and $B$ fails—if  $w$ triggers and complies with a more specific obligation overriding it. As a result, the $\neg a\wedge\neg f$-world  and the $a\wedge f$-world are ranked as equally good. %This is not to say that, under the proposed account, $\bigcirc\neg f$ will cease to hold. 

\section{Framework} \label{account}

%\subsection{The account}

In addition to a set of propositional letters, the Boolean connectives and the modal operator $\Box$, the language  $\mathcal{L}$ has (as primitives) a normality conditional operator $\Rightarrow $ and a conditional obligation operator $\dob(-/-)$. $\Rightarrow $ and $\dob(-/-)$ are read as before. %where the subscript $d$ stands fir ``defeasible''. 
For  $\Box A$, read ``It  is necessary that $A$". %For  $\dob (B/A)$, read  ``If $A$, then it ought to be the case that~$B$". 
For simplicity's sake, it is assumed that $\Rightarrow$ and $\Box$ do not occur within the scope of $\dob$ (and vice-versa) and that $\Box$ does not occur within the scope of $\Rightarrow$  (and vice-versa).  Iterations of $\bigcirc$ or $\Rightarrow$ are ruled out too.\footnote{
An example of a logic allowing nested occurrences of $\Rightarrow$ is \cite{lam91}'s system C4. 
Such a nesting is not directly relevant to the problem at hand (see, e.g., \cite{Makinson93} for a general discussion).} 
 As mentioned, $r$ (possibly indexed by a natural number)
is a schematic symbol that may stand for either a conditional obligation or a normality conditional.
Given some $r$, $h(r)$ is the head (or consequent) of $r$, and $b(r)$ is its body (or antecedent). As usual, $\Diamond A$ is short for $\neg\Box\neg A$. 
$\models_{\mathrm{S5}}$ denotes the entailment relation in modal logic S5 \cite{ddl:C80}, and $\models_{\PL}$  the entailment relation in (classical) propositional logic.

 %supplemented with 
% normality statements. These can be thought of a expressing ``soft'' constraints within a domain, allowing for exceptions.
%One can see these are %expressing ``soft'' constraints
%$\kbc$  is the subset  of rules %of the $\Rightarrow$-type, %and   $\kbo$  is the subset  of %rules of the $\bigcirc$-type.  %A set $\mathds{C}$ of hard information or constraints is assumed. $\mathds{C}$ contains alethic formulas, like $\Box(A\rightarrow\neg B)$, capturing the fact that $A$ and $B$ are jointly impossible.

%Jack cannot both rescue the child and meet Joe.

%Unlike Delgrande, I will require $h(r)$ be consistent.\footnote{Without this requirement, his theorem 4.1 does not go through. Put $\kb=\%%{A\rightarrow \bot\}$. A $\neg A$-world is strictly better than any $ A$-world, and thus minimal. } No other restriction other is placed on the %form of the rules.

\subsection{Preliminary Notions}

This section introduces notions needed for subsequent developments.
In what follows, $\kb=(\kbc,\kbo)$, where $\kbc$ and $\kbo$ are finite sets of  $\Rightarrow$-formulas and  $\bigcirc$-formulas, respectively. The formulas in $\kbc$ can be thought of as expressing ``soft'' constraints, allowing for exceptions. 
$\kbo$ is the normative system under consideration. Alethic formulas are used to encode hard (i.e., non-defeasible) information about the domain. For instance, $\Box(A\rightarrow\neg B)$
states that $A$ and $B$ 
cannot both hold.

\begin{definition}[Overridingness or defeat] \label{over:def}Given $r_{i},r_{j}\in \kbo$, and a set $\Gamma$ of alethic formulas, $r_{j}$ is said to override $r_{i}$ (notation: $r_{j}\triangleright r_{i}$) whenever 
\begin{description}
\item [{\rm (i)}] $\{h(r_{i}), h(r_{j})\}\cup \Gamma\models_{\mathrm{S5}}\bot$ 
\item [{\rm (ii)}] $b(r_{j})  \modpl b(r_{i})$ and $b(r_{i})\not\modpl b(r_{j})$ 
%($r_{j}$ is more %specific than %$r_{i}$)
\end{description}
%For $r_{i},r_{j}\in \kbo$, the following extra condition is added:%,  where $\dashv\vdash$ means that the two statements are interderivable:
\begin{description}
\item [{\rm (iii)}] $h(r_{i}), b(r_{j})\not\modpl \bot$.
%$b(r_{j})\;\;\cancel{\dashv\vdash} \;\;\neg h(r_{i})$,  ($r_{j}$ is not a CTD of $r_{i}$).
%$h(r_{i}), b(r_{j})\not\vdash\bot$  ($r_{j}$ is not a CTD of $r_{i}$).
%the heads of the rules are consistent)$
\end{description}
 
%$h(r)$ and  
%$h(r^{\prime})$ are classically inconsistent and $b(r)\vdash b(r^{\prime})$. 
%$C\Rightarrow  D \triangleright  A\Rightarrow  B$ (resp.
%$ \dob(D/C)  \triangleright \dob(B/A)$)
% whenever $A$ and $B$ are (classically) inconsistent and $C\vdash A$, where $\vdash$ is %the classical inference operation.
% For $C\Rightarrow  D \triangleright  A\Rightarrow  B$ (resp.
%$ \dob(D/C)  \triangleright \dob(B/A)$), read  $C\Rightarrow  D$ (resp.
%$ \dob(D/C)$) overrides $ A\Rightarrow  B$  (resp.
%$ \dob(B/A)$). 
\end{definition}

(i) and (ii)  may be found in \cite{Horty1993,Delgrande20}. (i) says that
the heads of the obligations are inconsistent in S5, given the hard information at our disposal. (ii) says that the overriding obligation is more specific. 
  (iii) is proposed by \cite{TvdT95}. It may be motivated as follows.
A contrary-to-duty (CTD) obligation says what should be done if another (primary) obligation is violated. (iii) tells us that  
$r_{j}$ is not a CTD of $r_{i}$. 
Consider \cite{for84}'s gentle murder paradox: \emph{Do not kill, but if you kill, do it gently!}. Without (iii), the CTD obligation \emph{If you kill, do it gently!} ($\bigcirc(k\wedge g/k)$) would override the primary obligation  \emph{Do not kill!} ($\bigcirc\neg k$).
%because killing gently and not killing are inconsistent$-$(i) is met$-$and killing entails $\top$$-$(ii) is met. 
This is usually considered counter-intuitive.

A model (referred to as an $\kb$-ordered model) has the form $$M=(W,\succeq_N, \succeq_I, v)$$ where $W$ is a (non-empty) set of worlds, $\succeq_N$ and $\succeq_I$ denote a normality and ideality (betterness) ordering on the worlds, and $v$ a valuation assigning to each world the set of propositional letters true in this world. ``$w\succeq_N v$" is read as ``$w$ is at least as normal as $v$", and ``$w\succeq_I v$" as ``$w$ is at least as good (or ideal) as $v$". ``$\succ_N$" (the strict counterpart of  ``$\succeq_N$") is defined by $w\succ_N v$ if and only if (iff) $w\succeq_N v$ and $v\not\succeq_N w$. ``$\succ_I$" (the strict counterpart of  ``$\succeq_I$") is defined analogously. Two worlds $w$ and $v$ are equally normal iff $w\succeq_N v$ and  $v\succeq_N w$. Given $X\subseteq W$, $\max_{\succeq_N} (X)$ is $\{w\in X: \forall v\in X\, (v\succeq_N w   \rightarrowtriangle w\succeq_N v)\}$. Equal goodness %and $\max_{\succeq_I}(X)$ are 
is defined analogously. $|\raisebox{.5ex}{.}|$ is the cardinality of $\raisebox{.5ex}{.}$.

Given a model $M$, and a world $w$,  the satisfaction relation $\models\subseteq W\times \mathcal{L}$ is defined inductively by ``$w\models p \mbox{ iff } w\in v(p)$'', with the usual clauses for the Boolean connectives. Those for the other connectives are introduced in the next sections.

Occasionally the focus will be on models that are replete in the sense of \cite{ddl:lou19}. Call a Boolean formula $A$ PL-consistent, if it is consistent in (classical) propositional logic, viz. $A\not\models_{\PL}\bot$.

\begin{definition}[Repleteness] \label{assm:rep} A model $M$ is replete if, whenever a Boolean formula $A$ is PL-consistent,  there is a world $w$ in $M$ such that $M,w\models A$. 
\end{definition}

%We will see in Example \ref{cott:ex2}) that it creates a problem when combining violation and exception. 

%Possible worlds are graded in terms of normality or ideality by comparing how well, or how badly, they comply with the conditionals. %they %tallying the number of conditionals they breach. 
%Higher rankings correspond to fewer violated conditionals, indicating a more normal or ideal state.
%The fewer conditionals are violated, the higher the ranking (indicating a more normal or ideal state). 
%However, a conditional overridden by a more specific conditional  is not considered in the count, as per Delgrande's proposal.
%Possible worlds are ranked by counting the number of conditionals they each violate. The less  conditionals are violated, the better (more normal or mode ideal) the world is. However, a conditional that is overridden by a triggered conditional does not count (Delgrande's main idea). 

\begin{definition}[Falsification set] Set  \begin{flalign*}
F(w)=&\{r_{i}\in\kbc: w\models b(r_{i}) \wedge\neg h(r_{i}) \}
%\mbox{ and } w \not\models b(r_{j})  \}
%\\ & \forall r_{j}\in\kbc
%\mbox{ s.t. }  r_{j}\triangleright  r_{i}\}
\end{flalign*}
$F (w)$ gathers all the conditionals in $\kbc$ falsified by $w$. 
\end{definition} 
\begin{definition}[Violation set]\label{violation:o} Set
\begin{flalign*}
V(w)=&\{r_{i}\in\kbo: w\models b(r_{i}) \wedge\neg h(r_{i})
\mbox{ and } w \not\models b(r_{j})  \\ & \forall r_{j}\in\kbo
\mbox{ s.t. }  r_{j}\triangleright  r_{i}\}
\end{flalign*}
$V(w)$ gathers all the obligations violated by $w$. The clause  ``$w \not\models b(r_{j})  \; \forall r_{j}\in\kbo
\mbox{ s.t. }  r_{j}\triangleright  r_{i}$'' ensures that $\bigcirc (B/A)$ is not considered violated by $w$—even when $A$ holds and $B$ fails—if  $w$ triggers and complies with a more specific obligation overriding it (cf. \S \ref{rank}). 
%an obligation whose antecedent  is true at $w$ and  consequent false at $w$ is not  %considered violated, if $w$ triggers and complies with a more specific obligation %overriding this obligation. 
\end{definition} 

%\begin{defi}[Violation set  2, ideality] Set
%\begin{flalign*}
%V_{o}(w, \kb)=&\{\dob(B/A)\in\kb: w \models A\wedge\neg B 
%\mbox{ and } w \not\models C \\ & \forall \dob(D/C) \in\kb
%\mbox{ s.t. }  \dob(D/C)  \triangleright \dob(B/A) \}
%\end{flalign*}
%\end{defi} 

%The next section deals with the normality component.

\subsection{Normality} \label{sec:norm}
%

%\subsection{Normality: lexicographic ordering}

%Some further notions are needed to define the %normality relation. 

%To define the normality relation $\succeq_N$, a few more concepts must be introduced.

This section describes   \cite{Lehmann95}'s lexicographic ordering method for normality.
Given $X\subseteq\kbc$, denote by $\m(X)$
%$\kbc_{m}$
the materialization of $X$, \ie 
$\{B\rightarrow C : B \Rightarrow C \in X\}$, where $\rightarrow$ is material implication. A Boolean formula $A$ is said to be exceptional for $X$ if $\m(X)\models_{\PL}\neg A$.
% \revv{EXPLAIN} and 
A conditional $A\Rightarrow B$ is said to be exceptional for
 $X$ if its antecedent $A$ is. Let $\varepsilon (X)$ be the set of  $A\Rightarrow B$s in $\kbc$ that are exceptional for $X\subseteq\kbc$. Define a sequence $\mathcal{E}_0$, ..., $\mathcal{E}_\infty$ of subsets of $\kbc$ as follows: 
 (i) $\mathcal{E}_0 = \kbc$; (ii) $\mathcal{E}_i = \varepsilon(\mathcal{E}_{i-1})$, for $0<i < m$ and 
 $\mathcal{E}_\infty = \mathcal{E}_m$, where $m$ is the smallest $k$ for which $\mathcal{E}_k = 
 \mathcal{E}_{l}$ for all $l>m$.  We call $m$ the order of $\kbc$. That $m$ always exists follows from the fact that
$\kbc$ is finite, and the following fact: 
 
\begin{restatable}{fact}{myfact}
	\label{lem:decreasing}  $(\mathcal{E}_i)_{i\geq 0}$ is (strictly) decreasing: 
	%$\mathcal{E}_0$, ..., $\mathcal{E}_m$ is a decreasing sequence of sets:
	$$\mathcal{E}_{i+1}\subset\mathcal{E}_{i} \text{ for all } i\geq 0$$
\end{restatable}

$(\mathcal{E}_i)_{i\geq 0}$ is called the LM-sequence (after \cite{LM92}). 

For technical reasons, an assumption of coherence of $\kbc$ is made. This assumption is intuitively plausible.

\begin{definition}[Coherence]\label{cons}
$\kbc$ is said to be coherent if $\neg \exists X\subseteq \kbc$ such that 
$\m (X)\models_{\PL} \bigwedge_{r\in X}\neg b(r)$, and incoherent otherwise. 
\end{definition}

\begin{assumption} \label{assm:cons} $\kbc$ is coherent. 
\end{assumption}

\noindent
%To justify my talk of consistency or satisfiability, note that it immediately follows that, if $\kbc$ is consistent, then there is a valuation %on the propositional atoms that makes $\wedge_{r\in\kbc}b(r)\wedge \wedge_{r\in\kbc}h(r)$ true
%$\wedge_{i=1}^{n}A_i \wedge \wedge_{i=1}^{n} B_i$ true. 
For instance,  $\{A\Rightarrow B, A\Rightarrow \neg B\}$ and $\{A\Rightarrow \neg A \}$ are incoherent, but $\{A\Rightarrow B, A\wedge C\Rightarrow \neg B\}$ is coherent.

%In this paper, I will assume that $\kbc$ is always consistent in the sense of Definition~\ref{cons}.

The following follows: 

%\begin{fact} \label{lem:cons} For all $r\in\kbc$, $b(r)\not\models\bot$.
%\end{fact}
%\begin{proof}Immediate. 
%\phantom\qedhere\end{proof}
\begin{fact} \label{lem:decreasing2} $\mathcal{E}_\infty =\emptyset$. 
\end{fact}

\begin{proof} Suppose, to reach a contradiction, that $\mathcal{E}_\infty = \mathcal{E}_m \not =\emptyset$, where $m$ is the smallest $k$ for which $\mathcal{E}_k = \mathcal{E}_{l}$ for all $l>m$.  Let $\mathcal{E}_m=\{r_1, ..., r_n\}\subseteq\kbc$. Since  $\mathcal{E}_{m}= \mathcal{E}_{m+1}$, 
$$\mathcal{E}_{m+1}=\{r_1, ..., r_n\}.$$
Since $\mathcal{E}_{m+1}= \varepsilon(\mathcal{E}_{m})$, 
$$\m (\mathcal{E}_{m})\models_{\PL}  \bigwedge_{i=1}^{n}\neg b(r_i)$$
As $\mathcal{E}_m\subseteq\kbc$, $\kbc$ is incoherent.% assump.~\ref{assm:cons}. 
%\vspace{-0.35cm}\phantom\qedhere
\end{proof}
%$\mathcal{E}_0$, ..., $\mathcal{E}_\infty$ is a decreasing sequence of sets.
 %  $\Gamma_{i+1}$ is the set of conditionals  in $\Gamma_{i}$ that are exceptional for $\Gamma_{i}$, (iii) $m$ is the %smallest  $i=1, ..., j$ such that
% $\Gamma_{i}=\Gamma_{i+1}=\emptyset$. 
% To partition $\kbc$ into levels of specificity, let us first define a decreasing sequence of sets of conditionals  $\Gamma_0 \supseteq \Gamma_1 \supseteq ...  \supseteq \Gamma_m$ by setting (i) $\Gamma_0 =\Gamma$, (ii), $\Gamma_{i+1}$ is the set of conditionals  in $\Gamma_{i}$ that are exceptional for $\Gamma_{i}$, (iii) $m$ is the smallest  $i=1, ..., j$ such that
%$\Gamma_{i}=\Gamma_{i+1}=\emptyset$. Call $m$ the order of $\Gamma$. 

\medskip
From this, we define the ranked partition of $\kbc$ as $\langle\Delta_{0}, ... , \Delta_{m}\rangle$ where
$\Delta_{i}=\mathcal{E}_{i}-\mathcal{E}_{i+1}$, for $i=0, ..., m-1$, and  $\Delta_{m}=\mathcal{E}_\infty$.
%Let us, next,  put $\Delta_{i}=\mathcal{E}_{i}-\mathcal{E}_{i+1}$, for $i=0, ..., m-1$, and let $\Delta_{\infty}$ be $\mathcal{E}%_\infty$. %The collection $\{\Delta_i \}$ defines a partition of $\kbc$. 
Intuitively, each $\Delta_{i}$ gathers the elements of $\kbc$ with the same level of specificity.
$\Delta_{0}$ is the set of less specific conditionals, $\Delta_{1}$ is the set of second-less specific ones, and so forth up to $\Delta_{m-1}$, the set of  most specific ones. Given some $\kbc$ of order $m$, to every subset $X\subseteq\kbc$ may be associated a tuple of natural numbers: $\langle
n_1,...., n_m \rangle$, where
%$n_i= |\Delta_{m-i} \cap  X|$,  for  all $i=0, ..., m$.
 %$n_0 = |\Delta_{\infty} \cap  X|$ and 
$n_i= |\Delta_{m-i} \cap  X|$,  for  all $i=1, ..., m$.  
These $X$s may be ordered using the
natural lexicographic ordering on their associated tuples.  Let the relation be denoted by $\gtrsim$. Its definition may be stated formally thus, where $\langle
n_1,...., n_m \rangle$ and $\langle
n^{\prime}_1,...., n^{\prime}_m \rangle$ are the tuples associated with $X$ and $Y$, respectively:\footnote{$n_0 = |\Delta_{m} \cap  X|$ does not need to be considered as a first coordinate, since under Assumption~\ref{assm:cons}  it is always equal to 0.}
\begin{flalign*}
X \gtrsim Y & \mbox{ iff } \langle
n_1,...., n_m \rangle \geq^{lex} \langle
n^{\prime}_1,...., n^{\prime}_m \rangle \\
& \phantom{ xxxx } \mbox{ iff }  n_i = n^{\prime}_i  \mbox{ for all } i \mbox{ such that  } 1\leq i\leq m\mbox{, or } \\
& \phantom{ xxxxxxxx }  n_i < n^{\prime}_i  \mbox{ for the first  }i \mbox{ such that  }  n_i \not= n^{\prime}_i
%& \mbox{ iff }  n_0 < n'_0  \mbox{ or } \\
%& \phantom{\mbox{ iff }}  n_0 = n'_0  \mbox{ and  }  n_1 \leq n'_1 \\
%n_i = n'_i \mbox{ for } i=0,1, ..., k-1 \mbox{ and }  n_k \leq n'_k
\end{flalign*}

The normality relation 
$\succeq_N$ can now be introduced. Intuitively,
a world is more normal if it falsifies fewer conditionals, giving precedence to more specific conditionals; falsifications of less specific conditionals are considered only in case of a tie on more specific conditionals.
%$w_1$ is at least as normal as $w_2$, if $w_1$ falsifies less conditionals than $w_2$ %does, taking into account the level of specificity of the conditionals involved. Formally,
\begin{definition}[Normality relation] \label{norm}Given a set $\kbc$  of normality conditionals, and $w_1,w_2\in W$, set
	% Define a sequence $N_1, ..., L_n,$ of normality levels in model 
	\begin{flalign*}
		w_1 \succeq_N w_2 &\mbox{ iff }  F(w_1)  \gtrsim F(w_2) 
	\end{flalign*}
	\end{definition}
 
\begin{example}[Exception, normality]\label{ex:norm}
	Let $\kbc=\{\top\Rightarrow \neg a, r\Rightarrow a\}$.
% Intuitively: normally, you do not eat asparagus; if you are at the restaurant, you normally eat asparagus. 
Fig.~\ref{table2} shows  the worlds' ranking, with tuples indicated by superscripts.
\end{example}
% The
% $r\wedge\neg a$-world (which falsifies the more specific conditional, but satisfies the more general one) is classified as less normal than the $r\wedge a$-world (which satisfies the more specific conditional, %but falsifies the more general one).  
%   Bo
% $r\wedge a$-world is classified as more normal than the  $r\wedge\neg a$-world.   Both falsifies exactly one conditional in $\kbc$, %but the former world verifies the more specific conditional, while the latter world does not.
	%one violating the more specific conditional (but satisfying the more general one).  Both violates the same number of conditionals, though. 
	
	\vspace{-0.5cm}
	\begin{center}
		\begin{figure}[h] \centering
			\includegraphics{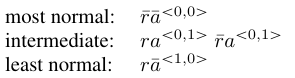}\caption{LEX ranking. 
	%		Any world that falsifies a more specific conditional is considered %even less normal than worlds falsifying only less specific conditionals.
		Falsifying a more specific conditional makes a world less normal
		 (Cf. \S \ref{rank}.)}\label{table2}
		\end{figure}
		\end{center}	

\vspace{-0.7cm}
\begin{remark} (i) $\succeq_N$ is reflexive, transitive and total (for all $w_1 , w_2 \in W$, 
$w_1 \succeq_N w_2$ or $w_2 \succeq_N w_1$); (ii)  $\succ_N$ is irreflexive, transitive and  modular (for all $w_1 , w_2, w_3 \in W$, if $w_1 \succ_N w_2$ then  $w_1 \succ_N w_3$ or  $w_3 \succ_N w_2$).
% and hence acyclic (for all $w_1 , ..., w_n \in W$, if $w_1\succ_N ... \succ_N w_n$, then $w_n\not\succ_N w_1$).  
%$\max_i =W$ if $i$ is odd, and $\max_i =\emptyset$ if $i$ is even.  
%(4) If $W$ is finite, then $\succeq_N$ is smooth, in the sense that if $w\models A$ then either $w\in\max_{\succeq_N} (\Vert A\Vert)$ or $\exists v \succ_N w$ s.t. $v\in\max_{\succeq_N} (\Vert A\Vert)$ (see e.g. \cite{DBLP:journals/jphil/Parent14}).
\end{remark}
$\Vert A\Vert^{M}$ is the truth-set of $A$ in $M$, viz., the set of worlds where $A$ holds. $M$ is omitted when no confusion can arise.
\begin{remark}\label{coll}
	If $\kbc=\emptyset$, then $\succeq_{N}=W\times W$ and $\max_{\succeq_{N}}(\Vert A\Vert)=\Vert A\Vert$.
\end{remark}

%
%When $\succeq_N$ is defined as in Remark \ref{tot:norm}, %I shall say that $\succeq_N$ is cardinality-based.
%
Infinite chains of increasingly more normal worlds are ruled out; $\succeq_N$ verifies the condition called ``smoothness'' by \cite{ddl:KLM90}  and ``stoppering'' by \cite{DBLP:conf/nmr/Makinson88}:

%\begin{proposition}[smoothness]\label{smoothn} If $w_1\models A$, then $w_1\in
%\max_{\succeq_N} (\Vert A\Vert)$ or $\exists w_2 \succ_N w_1$ s.t.  $w_2\in
%\max_{\succeq_N} (\Vert A\Vert)$.
%\end{proposition}

\begin{restatable}{proposition}{smoothness}\label{smoothn} If $w_1\models A$, then $w_1\in
\max_{\succeq_N} (\Vert A\Vert)$ or $\exists w_2 \succ_N w_1$ s.t.  $w_2\in
\max_{\succeq_N} (\Vert A\Vert)$.
\end{restatable}

 Intuitively, $\Box A$ holds in world $w$ when $A$ holds in all worlds, and $A\Rightarrow B$ holds in $w$ when  the most normal $A$-worlds are all $B$-worlds. Formally:

\begin{definition}[Truth-conditions, part I]  \par
\begin{flalign*}M,w\models \Box A    &\mbox{ iff }  W=\Vert A\Vert^{M} \\
M,w\models A\Rightarrow B    &\mbox{ iff }  \max_{\succeq_N} (\Vert A\Vert^{M})\subseteq \Vert B\Vert^{M}
\end{flalign*}
 %The truth conditions for $\bigcirc$ will be introduced separately.
\end{definition}

\begin{remark}[Systems S5 and R] (i) $\Box$ validates the axioms and rules of S5; (ii) $\Rightarrow$  validates the rules of \cite{LM92}'s system R (recast as a Hilbert-style axiom system).
%\begin{flalign}
%%& \mbox{ If } A\Rightarrow B\in \kbc \mbox{ then } %M,w\models A\Rightarrow B \label{inc}\tag{Inclusion}\\
%&\Box \mbox{ validates the axioms and rules of S5 } \label{s5}\tag{S5}\\
%\begin{split}
%&\Rightarrow \mbox{ validates the axioms and rules of \cite{LM92}'s }\\ & \mbox{system R} 
%\end{split}\label{p}\tag{R}
%\end{flalign}
\end{remark}
 %The labels in Fig.\,\ref{systemP} are drawn from the literature on nonmonotonic
% reasoning. (\ref{lle}), (\ref{rw}), (\ref{id}),  (\ref{cm}) and  (\ref{rm}) stand for ``left logical equivalence", ``right weakening", ``reflexivity'', ``cautious monotony", and ``rational monotony", respectively. 
%R is the flat fragment of \cite{Lew}'s system CV. 
 
%\begin{proof} For $\Box$, the verification is straightforward, and is omitted. For $\Rightarrow$, the reader is referred to  \cite{LM92}.
%\end{proof}

\subsection{Ideality} \label{sec:idea}

%This section deals with the deontic component.

%
%The idea is to restrict $\succeq_I$ (ideality ranking) to worlds within the same level of %normality. I will assume that worlds belonging to different levels of normality should not be %compared in terms of ideality. This hypothesis will play a crucial role in countering Horty's %objection to classical semantics: the fallacy of the prohibited exception.  

\begin{definition}[Ideality relation] \label{i:r} For a set $\kbo$ of obligations, and $w_1,w_2\in W$, set: 
% Define a sequence $N_1, ..., L_n,$ of normality levels in model 
\begin{flalign*}
w_1 \succeq_I w_2 & \mbox{ iff } %\exists i= 1, ...,  n \mbox{ s.t. } w_1 , w_2 \in \max_{i} \mbox{ and } 
V(w_1)\subseteq V(w_2)
%$V_{\Rightarrow}(w, \mathbf{R})=\{A\Rightarrow B\in\mathbf{R}: w\models A\wedge\neg B\}$.
%\item $V_{o}(w, \mathbf{R})=\{\do{B/A}\in\mathbf{R}: w \models A\wedge\neg B\}$.
\end{flalign*}
%The subscript $I$ stands for ``ideality". 
Intuitively, $w_1$ is at least as good as $w_2$ if the set of obligations violated by $w_1$  is included in the set of obligations violated by $w_2$.
%the idea is to restrict $\succeq_I$ to worlds within the same level of %normality only. 
%The slogan is: minimize the exceptions first, and then the violations.
\end{definition}

%\begin{remark}\label{rem} The alternative definition, putting $w_1 \succeq_I w_2$ whenever $|V(w_1)|\leq |V(w_2)|$, will also be considered. When this definition is assumed, I will describe $\succeq_I$ as %cardinality-based. 
%(i) If $\kbo=\emptyset$, then all the worlds are equally good; %(ii) If $\kbc=\emptyset$, then $\succeq_I$ is smooth.
%\end{remark}

\begin{remark}\label{rem2} $\succeq_I$ is reflexive, transitive and smooth.
	%within each level of normality (for all $i=1,2, ...$ and $w\in\max_i$, $w\succeq_I w$, and  
% for all $i=1,2, ...$ and $w_1 , w_2, w_3$ s.t. $w_1 , w_2, w_3\in\max_i$,  $w_1 \succeq_I w_2$ and $w_2 \succeq_I w_3$ imply $w_1 \succeq_I w_3$);
%\footnote{By construction, $\forall i,j.$  $ i\not= j \rightarrowtriangle \max_i \cap\max_j \not=\emptyset$. Therefore, the assumptions that $w_1 \succeq_I w_2$ and $w_2 \succeq_I w_3$ imply that the three worlds belong to the same level of normality.} %\todo{\scriptsize to show. How do I know the level of normality of $w_1$ exist? And that of $w_2$ is unique?}
% (ii) If $\succeq_I$ is cardinality-based, then $\succeq_I$ is in addition total 
% within each level of normality %
%(for all $w_1 , w_2 \in W$, 
%$w_1 \succeq_I w_2$ or $w_2 \succeq_I w_1$).
\end{remark}
%The truth conditions for $\Rightarrow$ may stay as they are. Intuitively, model $M$ satisfies $A\Rightarrow B$ when all the most plausible $A$-worlds are $B$-worlds.
%\begin{defi}[Truth-conditions, normality conditional]
%\begin{flalign}
%M\models A\Rightarrow B   & \mbox{ iff } \max_{\succeq_N} (\Vert %A\Vert)\subseteq \Vert B\Vert
%\end{flalign}
%\end{defi}
%
%\begin{prop} $\succeq_I$  is smooth (within each level).
%\end{prop}
%
%
%\begin{rem}\label{tot:norm} Let $|\raisebox{.5ex}{.}|$ denote the cardinality of $\raisebox{.5ex}{.}$. The possible worlds can also %be ordered by counting the elements in their violation set, putting $w_1\succeq_I w_2$ whenever   $|V_{o}(w_1 , \mathbf{R})|\leq |%V_{o}(w_2 , \mathbf{R})|$.  Thus defined, $\succeq_I$ is total (for all  $w_1, w_2$ in $W$, $w_1\succeq_I w_2$ or  $w_2\succeq_I %w_1$). When $\succeq_I$ is defined in this manner, $\succeq_I$ will be described as cardinality-based.
%\end{rem}
%\begin{proof} The argument is virtually the same as for Prop.  \ref{smoothn}.
%\end{proof}

%The new truth-conditions for $\dob$ are a variation on a standard method used in dyadic deontic logic to connect dyadic deontic logic %with the logic of preference. It is known that  $\bigcirc(B/A)$ can be defined by $\neg ((A\wedge \neg B) \geq (A\wedge B))$, with $%\geq$ 

The truth-conditions of $\dob$ are a variation on ones commonly used in a mono-preferential setting. To indicate what the variation is, a specific lifting relation must be introduced.\footnote{ \cite[p.\,4]{Halpern97} credits this lifting relation to \cite{L73}.} Below the superscript $s$ is for ``set''.

%  from worlds to a relation on sets of worlds. For an overview, see \cite{Halpern97}. 

%The framework  is flexible, and one may choose to work with its preferred one. I choose to work with the one proposed by \cite{L73}, %because IT COINCIDES WITH THE EVALUATION PATTERN IN TERMS OF BEST. 

%because it makes more clearly

%For the purpose at hand, it does not matter much which %way is chosen, and I will stick to the one proposed by %\cite{L73}. 
%Below the superscript $s$ is for ``set''. 

\begin{definition} \label{lifting} For  $U,V\subseteq W$, 
$U\succeq_I^{s} V$ (reading: $U$ is at least as good as $V$) iff $\forall v\in V$ $\exists u\in U$ s.t. $u\succeq_I v$. %$U\succ_I^{s} V$ is a shorthand for %$U\succeq_I^{s}V$ and  
%$V\not\succeq_I^{s} U$. 
\end{definition}

\begin{proposition} $\succeq_I^{s}$ is reflexive and transitive. % (ii) if $\succeq_I$ is cardinality-based, then  $\succeq_I^{s}$  is also total ($U\succeq_I^{s}V$  or  $V\succeq_I^{s}U$ for all $U,V\subseteq W$).
\end{proposition}
%\begin{proof} Straightforward. 
%\end{proof}

%Following \cite{DBLP:journals/jphil/MakinsonT01},  $\bigcirc (B/A)$ is said to have a consistent fulfillment, if $A\wedge B$ is PL-consistent. The following will come in handy:
%
%\begin{fa} Suppose each formula $\bigcirc (B/A)$ in $\kbo$ has a consistent fulfillment. If $w\in\max_{\succeq_I} (\Vert \top\Vert)$, then $V(w, \kb)=\emptyset$.  
%\end{fa} 
%
%\begin{proof} Suppose not. Then $\exists\bigcirc (B/A)\in \kbo$ s.t. $w\models A$ and $w\neg \models B$. By repleteness,
%  $\exists v$ s.t. $v\models A$ and $v\models B$. Since 
%
%\end{proof}

In deontic logic, it is usual to interpret  $\bigcirc (B/A)$  as  expressing a (strict) preference of  $A\wedge B$ over $A\wedge\neg B$. 
$\bigcirc (B/A)$ is equivalent with $A\wedge B > A\wedge\neg B$, where the connective  $>$ (on formulas) is defined by putting $A>B$ as true at a world iff
$\Vert B\Vert\not\succeq_I^{s} \Vert A\Vert$.
% (When $\succeq_I^{s}$ is total, it immediately follows $\Vert A\Vert\succeq_I^{s} \Vert B\Vert$.) 
%as meaning all the best $A$-worlds are $B$-worlds. Under suitable conditions this is equivalent to saying 
Intuitively, to determine the truth-value of  $\bigcirc (B/A)$ amounts to verifying whether  the set of  $A\wedge \neg B$-worlds is not at least as good as the set of $A\wedge B$-worlds, using an appropriate lifting relation. Instead of working with the set of worlds verifying the relevant formulas, this paper suggests working with the subset of those that are the most normal. 
%strictly better than the set of $A\wedge \neg B$-worlds, and hence compare under ideality  the two sets, using the lifting relation mentioned above. 
%The proposed variation consists in zooming on the subsets of their most plausible elements. 
%Rather than comparing the set of $A \wedge B$-worlds and the set of $A \wedge \neg B$-worlds, the idea is to zoom on the set of their most plausible elements. 
One gets: %the following evaluation rule. Roughly, it puts   
%It asserts that the most plausible worlds where both $A$ and $B$ are true are preferred over the most plausible worlds where $A$ is %true and $B$ is false. 
%Instead of saying the set of $A\wedge B$-worlds  is better than the set of $A\wedge\neg B$-worlds, one says that the set of most %plausible  $A\wedge B$-worlds is better than the set  of most plausible $A\wedge\neg B$-worlds.
%The following equivalent formulation will come in handy. Intuitively, it says that $U\succ_I^{s} V$ holds if the worst %elements in $U$ are all strictly better than eah best elements in $V$, or they are equally good. In other worlds, %%$U\succ_I^{s} V$ fails, when a best $V$ is strictly better than a worst $U$ (or they are incomparable). 
%\begin{prop}[min-max] If $\min_{\succeq_I} (U)\not=\emptyset$, then 
%\begin{flalign}
%U\succ_I^{s} V \Leftrightarrow \forall a\in \min_{\succeq_I} (U)  \forall b\in \max_{\succeq_I} (V) a\succ b \mbox{ %or } 
%a\approx b
%\end{flalign}
%\end{prop}
%\begin{proof} ($\Rightarrow$) Suppose $U\succ_I^{s} V$. Let $a\in \min_{\succeq_I} (U)$  and  $b\in %%\max_{\succeq_I} (V)$. 
%\end{proof}
 $\dob (B/A)$ is true, whenever the set of most normal $A\wedge \neg B$-worlds is not at least as good as the set of most normal $A\wedge B$-worlds. 
%most normal $A\wedge B$-worlds is strictly better than the set of most normal $A\wedge \neg B$-worlds. 
This kind of variant rule, where normality and ideality are combined, is familiar from the deontic logic literature$-$see e.g. \cite{jackson84,vtt97}. %When unfold, the truth-conditions for $\dob (B/A)$ read: 

%See XYZ. Formally:
\begin{definition}[Truth-conditions, part II] \label{def:o}
\begin{flalign*}
w\models\dob (B/A)   \mbox{ iff } 
%\max_{\succeq_N} (\Vert A\wedge B\Vert)  \succ_I^{s}
%\max_{\succeq_N} (\Vert A\wedge \neg B\Vert) \\
%\mbox{ iff }
%\max_{\succeq_N} (\Vert A\wedge B\Vert)  \succeq_I^{s}
%\max_{\succeq_N} (\Vert A\wedge \neg B\Vert)  \mbox{ and }\\
& \max_{\succeq_N} (\Vert A\wedge \neg B\Vert)  \not\succeq_I^{s}
\max_{\succeq_N} (\Vert A\wedge B\Vert)  
%\\
%\mbox{ iff } & 
% \forall v \in 
%\max_{\succeq_N} (\Vert A\wedge \neg B\Vert)  \;\; \exists w \in 
%\max_{\succeq_N} (\Vert A\wedge B\Vert) \;\; w\succeq_I v  \mbox{ and }\\
%&
%\exists v \in 
%\max_{\succeq_N} (\Vert A\wedge B\Vert) \;\; \forall w \in 
%\max_{\succeq_N} (\Vert A\wedge \neg B\Vert) \;\;w\not\succeq_I v
\end{flalign*}
\end{definition} 

Taken together, Def.\;\ref{lifting} and \ref{def:o} yield the following derived truth-conditions. They state that $\dob (B/A) $ holds, if there is a most normal $A\wedge B$-world that is not worse than (or weakly dominated by) any most normal $A\wedge \neg B$-world.

\begin{proposition}[Derived truth-conditions] \label{def2:o}
\begin{flalign*}
w\models\dob (B/A)   %\mbox{ iff } 
%\max_{\succeq_N} (\Vert A\wedge B\Vert)  \succ_I^{s}
%\max_{\succeq_N} (\Vert A\wedge \neg B\Vert) \\
%\mbox{ iff }
%\max_{\succeq_N} (\Vert A\wedge B\Vert)  \succeq_I^{s}
%\max_{\succeq_N} (\Vert A\wedge \neg B\Vert)  \mbox{ and }\\
%& \max_{\succeq_N} (\Vert A\wedge \neg B\Vert)  \not\succeq_I^{s}
%\max_{\succeq_N} (\Vert A\wedge B\Vert)  
%\\
\mbox{ iff } & 
% \forall v \in 
%\max_{\succeq_N} (\Vert A\wedge \nefg B\Vert)  \;\; \exists w \in 
%\max_{\succeq_N} (\Vert A\wedge B\Vert) \;\; w\succeq_I v  \mbox{ and }\\
%&
\exists v \in 
\max_{\succeq_N} (\Vert A\wedge B\Vert) \;\; \\
& \phantom{xxxxxx}\forall u \in 
\max_{\succeq_N} (\Vert A\wedge \neg B\Vert) \;\;u\not\succeq_I v
\end{flalign*}
\end{proposition} 

% {\sc ToDo: explain the restriction to the most plausible. the states known to be impossible are
%irrelevant
%from the poin}
 
%In Section \ref{benchmark}, we will see that the restriction to the most normal worlds will play a role in the elimination of the drowning effect. The assumption is that, when one utters %``It ought to be that $B$, given $A$", one assumes that things are as normal as possible. Thus, exceptional worlds are not relevant for the truth of the obligation. 

A few words about zooming in on the most normal worlds. The assumption is that when I say ``it ought to be that 
$B$, given $A$," I typically do so under the presumption that things are as normal as possible. In other words, exceptional or abnormal worlds  are treated as irrelevant when assessing whether the obligation holds. This assumption is in itself plausible. As shown in \S \ref{benchmark}, it prevents the drowning effect.

It is now possible to  explain why Def.\;\ref{lifting} was chosen. When $\kbc=\emptyset$, 
the truth-conditions for $\bigcirc$ coincide with those employed in deontic logic to accommodate unresolved conflicts, namely situations in which both $\bigcirc A$ and $\bigcirc \neg A$ hold without either taking precedence \cite{ddl:vanFraassen1972,ddl:Lew74,ddl:PS97}. This pattern of evaluation may be referred to as the $\exists\forall$ truth-conditions. The use of this approach to accommode conflicts is advocated by \cite{Gob03,ddl:gob13}.  The rule reads: $\bigcirc (B/A)$ is true in $w$ if there is some $A\wedge B$-world such that, as we go up in the ordering, the material implication $A\rightarrow B$ always holds.\footnote{
Compared to the standard evaluation pattern in terms of ``best'' \cite{H69}, the main advantage of the 
 $\exists\forall$  rule is that it can accommodate  conflicts without generating a deontic collapse. In particular, the set $\{\bigcirc A,\bigcirc \neg A\}$ 
is satisfiable in a model. But, unlike the  best-based evaluation rule, the 
 $\exists\forall$  rule does not warrant the inference from $\bigcirc A$ and $\bigcirc \neg A$ to $\bigcirc B$ for an arbitrary $B$, implying the collapse of all normative distinctions in the face of a conflict. }
%which arises in Hansson’s best-worlds semantics when normative conflicts collapse the set of ideal worlds.
%over the standard evaluation pattern in terms of ``best'' \cite{H69} is that it %accommodates (unresolved) conflicts between obligations without resulting in deontic %explosion. That is, %$\bigcirc A$ and $\bigcirc \neg A$ can bot
%h be satisfiable in a model where the betterness relation is not total, and thus they remain logically consistent. Crucially, unlike the best-based evaluation, 
 Formally, 
\begin{proposition}[$\exists\forall$ truth-conditions]\label{thm:collapse}
%\begin{restatable}[$\exists\forall$ rule]{thm}{goldbach}
%\label{thm:collapse}
 Let $\kbc=\emptyset$.  %, and $\succeq_I$  is cardinality-based. 
 %Suppose  also that  $\Vert A\Vert\not=\emptyset$. 
 % $\max_{\succeq_I} (\Vert A\Vert)\not=\emptyset$. 
 Then:
\begin{flalign*}
%\begin{split}
w\models \dob (B/A)   & \mbox{ iff } \exists v  \bigl( v\models A\wedge B  \mbox{ and }\\
& \phantom{xxxx} \forall u   \,(u\succeq_I v  \rightarrowtriangle  u\models A \rightarrow B)\bigr)
%\end{split}\tag{$\exists\forall$}\label{ea-rule}
\end{flalign*}  
%\end{restatable}
\end{proposition}
\begin{proof}
This follows at once from Prop. \ref{def2:o} and Rem. \ref{coll}.%\phantom\qedhere
\end{proof}
$\bigcirc$ as defined in Prop. \ref{thm:collapse} satisfies the axioms and rules of  \cite{Gob03}'s system DP.\footnote{DP differs from \cite{L73}'s system VN in not containing the axiom $\bigcirc (B/A)\rightarrow \neg\bigcirc (\neg B/A)$ nor the principle AND, which is present in \cite{ddl:KLM90}'s system\;P. }

Looking back, this consideration is what primarily motivated the use of Def.\;\ref{lifting}. It is reasonable to expect that an established theory of conditional obligation should emerge as a limiting case when normality considerations are set aside.

\begin{comment}
 I will refer to this evaluation rule as the (\ref{ea-rule}) rule. As explained in Section \ref{nml}, the main advantage of the (\ref{ea-rule}) rule over  \cite{H69}’s truth-conditions in  terms of ``best'' is that it handles unresolved  conflicts between obligations more satisfactorily.
%\footnote{
%\revvv{The main advantage of the $\exists\forall$ rule over the standard evaluation pattern in terms of ``best'' \cite{H69} is that it allows for conflicts between obligations without resulting in deontic explosion. $\bigcirc A$ and $\bigcirc \neg A$ can both be satisfiable in a model where the betterness relation is not total, and thus they remain logically consistent. Crucially, unlike the best-based evaluation, $\bigcirc A$ and $\bigcirc \neg A$ do not entail $\bigcirc B$ for an arbitrary $B$.}
Looking back, this consideration is what primarily motivated the use of Definition \ref{lifting}. It is reasonable to expect that a widely accepted theory of conditional obligation should emerge as a limiting case when normality considerations are set aside.
\end{comment}

\subsection{Nonmonotonic Layer} \label{nml}

% is defined the way 

%My consequence relation $\Vdash$ is the modal analog of 
%The consequence or entailment relation of DP as %defined in \cite{Gob03} is monotonic. 
 This section adds a nonmonotonic layer to the framework.

It is helpful to adopt Horty’s terminology. His entailment  relation $\mid\sim$ operates on a theory $\Delta$, and outputs formulas (\eg obligations)  ``entailed'' by the theory (see, e.g., \cite[p. 80]{horty12}).
%I will stick to his terminology. Call $\Delta= (\Gamma, \kb)$ a theory
%To make this explicit, 
%and use his entailment relation $\mid\sim$. 
In what follows, $\Gamma$ is a (finite) set of ordinary (Boolean or alethic) formulas. 
%propositions, including a set $ \mathds{C}$ of hard information,\footnote{$\mathds{C}$ is the set of alethic formulas appearing in Definition \ref{over:def} (which helps to determine which rule overrides the other).} 
A theory  $\Delta$ is a pair $(\Gamma, \kb)$, where $\kb=(\kbc,\kbo)$.
$\mid\sim$ is defined as in~\cite{ddl:PS97}. 
The definition refers to  the class $\mathcal{C}(\kb)$ of all $\kb$-ordered models, whose normality and ideality orderings are determined by $\kb$ (\emph{as per} \S \ref{sec:norm} and \S \ref{sec:idea}). Accordingly, one can say that $\mid\sim$ is (indirectly) parametrized by $\kb$.
\begin{definition}\label{cons}% $\Gamma\Vdash_{\kb} A$ 
	$\Delta \mid\sim A$ (``$\Delta$ entails $A$'') iff  $\Gamma \models_{\mathcal{C}(\kb)} A$.
	%	iff $\Gamma\models_{\mathcal{C}(\kb)} A$.
	%, where $ \kb\subseteq\Gamma$. 
	%wrt to the class of models where the violation sets inducing the ordering are determined by $\kb \subseteq\Gamma$.
\end{definition}
% We say that $\Delta$ entails a formula $A$ (notation: $\Delta \mid\sim A$) if $\Gamma \models_{\mathcal{C}(\kb)} A$.
Intuitively, Def. \ref{cons} says that theory $\Delta$ entails $A$ if 
in any  $\kb$-ordered model $M$, and any world $w$ in $M$, if all the formulas in $\Gamma$ hold in $w$, then so does $A$.  
Def.\;\ref{cons} does not  require that every normality conditional $A \Rightarrow B$ in $\kbc$ be satisfied in world $w$. Prop. \ref{th:inc} will demonstrate that this requirement is automatically met. Similarly, the definition imposes no such requirement on the obligations $\bigcirc(B/A)$ in $\kbo$; Prop.~\ref{inc:cor} and~\ref{no-drow} specify when this too holds automatically.

% in the model under suitable conditions, for instance
%if $\kbo$ does not contain a more specific, overriding %obligation.

To capture the property of monotonicity  in our setting, the usual notion of containment between theories must be redefined.

\begin{definition}\label{cont}
	 $\Delta_1 = (\Gamma_1, \kb_1)$ is contained in $\Delta_2 = (\Gamma_2, \kb_2)$ (notation: $\Delta_1 \sqsubseteq \Delta_2$) whenever %$\mathds{C}_1 \subseteq \mathds{C}_2$, 
	 $\Gamma_1 \subseteq \Gamma_2$, $\kbc_1 \subseteq \kbc_2$,  and $\kbo_1 \subseteq \kbo_2$. 
	\end{definition}
%Containment between theories is %defined as follows: $\Delta_1 = %(\Gamma_1, \kb_1)$ is contained in %$\Delta_2 = (\Gamma_2, \kb_2)$ %(notation: $\Delta_1 \sqsubseteq %\Delta_2$) whenever $\mathds{C}_1 %\subseteq \mathds{C}_2$, $\Gamma_1 %\subseteq \Gamma_2$ and $\kb_1 %\subseteq \kb_2$.
 Def.\,\ref{cont} means that $\Delta_2$ contains all the facts, hard information, explicit obligations and explicit normality conditionals  already present in $\Delta_1$. Monotonicity is the property that if $\Delta_1 \mid\sim A$ and $\Delta_1 \sqsubseteq \Delta_2$, it follows that $\Delta_2 \mid\sim A$.

% When convenient, I will use Prakken and Sergot's compact ``hat" notation
%$\widehat{A}$ to indicate that $A\in\kb$. No hat will then mean that  $A\in\Gamma$.
%
%Following common practice, an obligation in $\kbo$ can be interpreted as a prima facie obligation. An obligation outputted by the model can be interpreted as
%%$A$ at the right-hand side of $\models$ will typically be a deontic formula $\bigcirc(B/A)$. This one must be interpreted as
% an ``all things considered'' obligation, resulting from after weighing all prima facie obligations against each other. Not all prima facie obligation becomes an all things considered obligation. 
% %T
% %his is why inclusion holds in a restricted form only (Corollary \ref{inc:cor}). 
% The distinction between prima facie and all things considered obligation will gain greater significance in Section \ref{priorities}, when a priority relation on obligations is introduced. 

\subsection{Example} \label{benchmark}

A specificity structure is represented as $\{\bigcirc A$, $\bigcirc (\neg A/B),\top\Rightarrow\neg B\}$.  On the technical side, $\top\Rightarrow\neg B$ is used to generate a normality ordering in the model. Conceptually, $\top\Rightarrow\neg B$ fixes  $\neg B$ as the default context. Since $\neg B$ is declared holding by default, the default obligation in that context is $\bigcirc A$, defeasibly overridden by  $\bigcirc (\neg A/B)$.

%To %illustrate further the interaction between normality and ideality, I consider an extended version of the asparagus scenario, obtained by  adding the premise ``If you are at the restaurant, you normally eat asparagus".
%\footnote{This addition was suggested by an anonymous referee.}

\begin{example}\label{horty}[The asparagus, extended version] Consider a theory $\Delta =(\Gamma,\kb)$ where $\Gamma=\emptyset$ and $\kbc=\{ \top\Rightarrow\neg \mathit{a}, r\Rightarrow\mathit{a}\}$ and  
$\kbo=\{\dob\neg \mathit{f}, \dob(\mathit{f}/\mathit{a}), \dob\mathit{n}\}$.
%\begin{flalign*}
%\dob\neg \mathit{f}\qquad \dob(\mathit{f}/\mathit{a}) \qquad\dob\mathit{n} \qquad \top\Rightarrow\neg \mathit{a}  \qquad  r\Rightarrow\mathit{a}
%\end{flalign*}
%The conditional ``$\top\widehat{\Rightarrow}\neg\mathit{a}$'' tells us that being served asparagus is exceptional. As before,
Given $\kbc$, the set of all possible worlds is partitioned into three
equivalence classes: the class of those with  $\neg a$ (=the most normal worlds), the class of those with $r\wedge a$ or $\neg r\wedge a$ (=the second most normal worlds) and 
the class of those with $r\wedge \neg a$ (=the least normal worlds). Fig. \ref{aspa} shows the resulting $\kb$-ordered model $M$.
The ``$y$-axis" represents normality. The higher up a world is, the more normal it is. The ``$x$-axis" represents ideality. The more to the right a world is, the farther away  to full ideality it is.
%\footnote{The labels ``best'', ``second-best'', and ``worst'' are used in a broad sense. Worlds at the same level of ideality may be incomparable, since the ordering is given by set-theoretic inclusion.} 
% The (rotated) symbol ``='' indicates that the related worlds are equally good. 
%We distinguish between ideal worlds, sub-ideal worlds, and worst worlds.
%, and the symbol `$\not=$" indicates that they are incomparable under ideality. 
% A world where $\neg r\wedge\neg a\wedge \neg f\wedge n$ holds is represented as $\bar{r}\bar{a}\bar{f}n$. 
%the less ideal it is.
%from lower to higher in both cases. axe represents ideality and the vertical axe %represents normality:
The account gives the expected solutions. First, each of $\dob (\mathit{f}/\mathit{a})$, 
$\dob\mathit{n}$, $\top\Rightarrow\neg \mathit{a}$ and $r\Rightarrow\mathit{a}$ holds in a given world $w$ in $M$, and so each is defeasibly entailed by $\Delta$. 
Second, we have that 
%More generally, we will see that the inclusion requirement is met. Second, we have that, where $w$ represents an arbitrarily chosen world, 

\vspace{-0.2cm}
%\smallskip
\begin{minipage}{3cm}
\setlength{\parindent}{2em}
\begin{flalign}
& \Delta\mid\sim \dob (\mathit{n}/a)\label{sa1}\\
& \Delta \bcancel{\mid\sim} \dob (\neg \mathit{f}/a)\label{sa2}
\end{flalign}
\end{minipage}
\begin{minipage}{5cm}
\setlength{\parindent}{2em}
\begin{flalign}
& \Delta \bcancel{\mid\sim} \dob \neg \mathit{a}  \mbox{ (witness: } \mathit{afn}\mbox{)} \label{kratz}\\
& \Delta\mid\sim  \dob\neg\mathit{f} \mbox{ (witness: } \mathit{\bar{r}\bar{a}\bar{f}n}\mbox{ )}\label{del}
\end{flalign}
\end{minipage}

\medskip
\noindent
(\ref{sa1}) and  (\ref{sa2}) show that the ``right'' amount of strengthening is obtained. 
%of the antecedent is warranted in the absence of evidence to the contrary, yet it is blocked when a more specific obligation supersedes the obligation. 
% Let us $A\succ_I^{s}B$ abbreviates $\Vert A\Vert\succ_I^{s}\Vert %B\Vert$. 
%We have (\ref{sa1}), because 
%$\max_{\succeq_N}(\Vert  a\wedge\neg\mathit{n}\Vert)  \not\succeq_I^{s}\max_{\succeq_N}(\Vert a\wedge \mathit{n}\Vert)$. We have (\ref{sa2}), because 
%$\max_{\succeq_N}(\Vert a\wedge\mathit{f}\Vert) \succeq_I^{s}\max_{\succeq_N}(\Vert  a\wedge\neg\mathit{f}\Vert)$.
 %The worst most plausible $an$ is still as good as the best most plausible $a\bar{n}$, but the worst most plausible $a\bar{f}$ is worse than the best most plausible $af$. 
% Semantic explanation: $n$ and $f$ are best among the set of $a$-worlds, the subset of exceptional worlds.
%\todo{the best of the most plausible $a\bar{n}$ is not strictly better than the worst most plausible $an$. }
% the most normal $a\wedge n$-worlds
%and the most plausible  $a\wedge f$-worlds are both exceptional worlds; the first arez %better than the second; they also turn out to make $f$ true. overall are better than the %most plausible  $a\wedge\neg n$-worlds,  and than the most plausible 
%(\ref{kratz}) is the case, because $\max_{\succeq_N}(\Vert \mathit{a}\Vert)\succeq_I^{s}\max_{\succeq_N}(\Vert \neg \mathit{a}\Vert) $. Therefore, 
 (\ref{kratz}) and (\ref{del})  show that the fallacy of the prohibited exception and the drowning effect are avoided. %This is because worlds belonging to different levels of normality are incomparable %under the ideality relation. 
%Worlds belonging to different levels of normality are %incomparable under $\succeq_I$. $\max_{\succeq_N}(\Vert \neg\mathit{a}\Vert) \not\succeq_I^{s}\max_{\succeq_N}(\Vert  \mathit{a}\Vert) $, and so
%$\max_{\succeq_N}(\Vert \neg\mathit{a}\Vert) %\not\succ_I^{s}\max_{\succeq_N}(\Vert  \mathit{a}\Vert) %$.
%$\geq_{I}$. 
%Normality and ideality are two distinct concepts. 
%(\ref{del}) shows that the drowning %effect does not apply.

%the prohibition of $\mathit{f}$ still holds in context $\top$. Indeed, $\max_{\succeq_N}(\Vert \mathit{f}\Vert) \not\succeq_I^{s}\max_{\succeq_N}(\Vert  \neg\mathit{f}\Vert)$.

\begin{figure}[h]\centering
\includegraphics[scale=0.3]{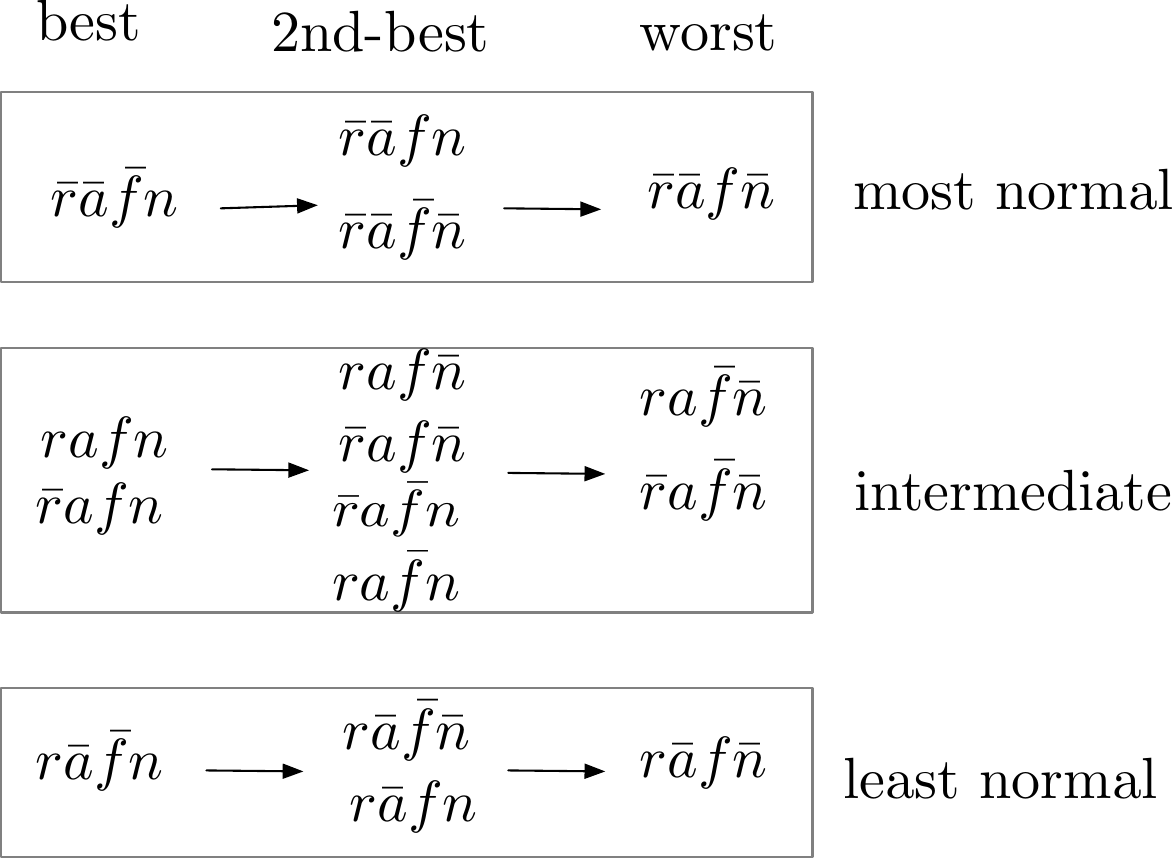}\caption{The asparagus}\label{aspa}
\end{figure}

%The two accounts fail in as much as they do not keep them apart. 
%and that the problem faced by Delgrande's account is avoided. 
%Pne does not compare 
% The account does not transform the exceptional into a prohibition.  This is essential %because worlds in different levels of normality remain incomparable under the ideality %relation. (\ref{del}) tells us that t We do not ``loose'' obligations.  
\end{example}
We can now see why zooming in on the most normal worlds in Def.\,\ref{def:o} prevents the drowning effect.  As mentioned, the idea is that, when I say e.g. ``you ought not to eat with your fingers,'' I do so assuming that circumstances are as normal as possible. Abnormal or exceptional worlds—such as those where you are served asparagus—are disregarded when evaluating whether the obligation holds. Fig.~\ref{aspa} illustrates that, if such worlds were included, $\bigcirc\neg f$ would be false. For $\Vert \mathit{f}\Vert \succeq_I^{s}\Vert \neg\mathit{f}\Vert$ while $\max_{\succeq_N}(\Vert \mathit{f}\Vert) \not\succeq_I^{s}\max_{\succeq_N}(\Vert \neg\mathit{f}\Vert)$.
Thus, obligations do not drown because normality acts as an extra filter. 
%(cf. \S \ref{rank})

In this setting, the fallacy of the prohibited exception does not arise because worlds at different levels of normality remain comparable under the ideality relation.  As mentioned in \S \ref{rank}, the key point is that,
other things being equal with respect to further obligations, a world in which the general obligation is violated and the more specific one fulfilled ($a\wedge f$) is as good as a world in which the general obligation is fulfilled ($\neg a\wedge \neg f$). 
This remains true when focusing on the most normal $a\wedge f$- and $\neg a\wedge \neg f$-worlds, even though they lie at different levels of normality.

\section{Properties} \label{laws:section}

This section generalizes the points made in the previous example, by identifying postulates characteristic of the system at the nonmonotonic layer. They serve as laws for the explicitly given obligations and normality conditionals.

The first postulate is the principle of inclusion,\footnote{The name is from \cite{CasiniMV19}.} stating that
explicit conditionals are defeasibly entailed.
%For obligation, inclusion holds only in a restricted form.
For normality conditionals, inclusion holds without restriction (subject to a consistency proviso).
%\begin{makethm}[Inclusion, normality]{thm}{inclu:cond}
%\label{cond:inclusion}

%\begin{makethm}[Inclusion]\label%{cond:inclusion}

%\newcommand{\definemythm}{
\begin{restatable}[Inclusion, normality]{proposition}{mythm}
	\label{th:inc}
	Consider a theory $\Delta$ such that $A\Rightarrow B\in\kbc$   and $A\wedge\neg B$ is PL-consistent. Then (under the assumption of repleteness)
	\begin{align*}
		&  \Delta\mid\sim A\Rightarrow B %\label{thm:inclusionin} 
	\end{align*}
\end{restatable}

The property of inclusion for obligation holds only in a qualified form, as a direct consequence of the fact that strengthening of the antecedent (SA) holds as a defeasible principle.

%The next property is 
%strengthening of the antecedent (SA). It holds as a defeasible principle. 

%A defeasible form of antecedent strengthening %is available in the following form:

\begin{proposition}[Defeasible SA, obligation]\label{laws}
	Suppose $\Delta$ is such that $\dob (B/A)\in\kbo$,
	$\neg\exists r\in\kbo$ \mbox{ s.t. } $r \triangleright \dob (B/A)$ and
	$\Diamond (A \wedge B\wedge C) \in\Gamma$. Then,  
		\begin{equation*}
	\Delta\mid\sim \dob (B/A\wedge C) 
	\label{def:sa:o} %\tag{def-SA}
	\end{equation*}

%	\begin{equation}
%		\begin{split}
%		& \mbox{ If } \dob\!(B/A)\in\kbo,  \neg\exists r\in\kbo \mbox{ s.t. } r \triangleright \dob (B/A) \mbox{ and }  \notag \\ &  \qquad \Diamond (A \wedge B\wedge C) \in\mathds{C} \subseteq\Gamma
%		 \mbox{, then }  \Gamma \Vdash_{\kb}\dob (B/A\wedge C) 
%		\end{split}\label{def:sa:o} \tag{def-SA}
%	\end{equation}
\end{proposition}
\begin{proof}  We show that, for all ${\bf R}$-ordered model $M$, and all $w$ in $M$, if $w\models\bigwedge\Gamma$, then 
$w\models\dob (B/A\wedge C)$. As before, $\exists u$ s.t. $u\models A\wedge B\wedge C$.  By smoothness, $\exists v\in\max_{\succeq_N} (\Vert A\wedge B\wedge C\Vert)$.
	Let $z \in 
	\max_{\succeq_N} (\Vert A\wedge  C\wedge\neg B\Vert)$. 
	We have $v\models  A\wedge B$  and  $z\models  A\wedge \neg B$. Since $\neg\exists r\in\kbo \mbox{ s.t. } r \triangleright \dob (B/A)$, $\dob (B/A)\in V(z)$. But $\dob (B/A)\not\in V(v)$. Hence $z\not\succeq_I v$. %\vspace{-0.25cm}\phantom\qedhere
\end{proof}
The case where $C:=\top$ yields inclusion. This property holds under the assumption that the normative system contains no more specific obligation overriding the one in question. This qualification is justified, since a more specific obligation may override the general obligation in that context. This condition is sufficient, though not necessary. 
%Analogously, Prop. \ref{no-drow} provides sufficient conditions under which an overridden obligation does not ``drown.''

%Before considering inclusion for obligation 
\begin{proposition}[Inclusion, obligation]\label{inc:cor}
	Suppose $\Delta$ is such that $\dob (B/A)\in\kbo$, $\neg\exists r\in\kbo \mbox{ s.t. }  r \triangleright \dob (B/A)$ and    $\Diamond (A \wedge B) \in\Gamma$. Then  
	\begin{flalign*}\Delta\mid\sim\dob (B/A) %\label{inc:o:cor}\notag  %
%	%\tag{$\bigcirc$-%Inc} 
		\end{flalign*}
%	%	If  $\dob(B/A)\in\kbo$, $\neg\exists r\in\kbo$ s.t.  $r \triangleright \dob (B/A)$ and  $A\wedge B$ is PL-consistent, then
%	%\mbox{ and } \Diamond (A\wedge B)\in\Gamma
%%	\begin{flalign*}
%%		\begin{split}
%%			\mbox{If  }\dob\!(B/A)\in\kbo, \neg\exists r\in\kbo \mbox{ s.t. }  r \triangleright \dob (B/A)  \mbox{ and  }  \Diamond (A \wedge B) \in\mathds{C} \subseteq\Gamma, \mbox{ then } \\
%%			  \Gamma \Vdash_{\kb} \dob (B/A) %\tag{$\bigcirc$-Inc} \label{inc:o:cor} 
%%			  \end{split}
%%				%&\mbox{ If } A\Rightarrow B\in\kbc,  \mbox{ then } \Gamma\Vdash A\Rightarrow B \label{inc:i:cor} \\
%%	\end{flalign*}
\end{proposition}

Nonmonotonicity follows from Prop. \ref{laws} and Ex. \ref{horty}. 

%To make this explicit, it is convenient to adopt Horty’s terminology and his entailment relation $\mid\sim$. Consider a pair $(\Gamma, \kb)$ as the analogue of what Horty refers to as a theory, denoted $\Delta$. We say that $\Delta$ entails a formula $A$ (notation: $\Delta \mid\sim A$) if $\Gamma \models_{\mathcal{C}(\kb)} A$.
%Containment between theories is defined as follows: $\Delta_1 = (\Gamma_1, \kb_1)$ is contained in $\Delta_2 = (\Gamma_2, \kb_2)$ (notation: $\Delta_1 \sqsubseteq \Delta_2$) whenever $\mathds{C}_1 \subseteq \mathds{C}_2$, $\Gamma_1 \subseteq \Gamma_2$ and $\kb_1 \subseteq \kb_2$. Intuitively, this means that $\Delta_2$ incorporates all the facts, explicit obligations, explicit normality conditionals, and hard information already present in $\Delta_1$. Monotonicity is then the property that if $\Delta_1 \mid\sim A$ and $\Delta_1 \sqsubseteq \Delta_2$, it follows that $\Delta_2 \mid\sim A$. To establish nonmonotonicity, it suffices to exhibit a counterexample.

\begin{example}[Nonmonotonicity] Consider $\Delta_1=(\Gamma_1, \kb_1)$, with $\Gamma_1=\{\Diamond (a\wedge \neg f)\}$ and $\kbo_1=\{\bigcirc \neg f\}$. By Prop.~\ref{laws}, $\Delta_1\mid\sim\bigcirc (\neg f/a)$. Let $\Delta_2$ be as in Ex. \ref{horty}, except that 
we now set $\Gamma_2=\{\Diamond (a\wedge \neg f)\}$. 
	%	$=(\Gamma_2, \kb_2)$, with $\Gamma_2=\{\Diamond (a\wedge \neg f)\}$ and $\kb_2=\{\bigcirc \neg f, \bigcirc (f/a),\top\Rightarrow\neg a\}$. 
	$\Delta_1 \sqsubseteq \Delta_2$, but $\Delta_2\bcancel{\mid\sim}\bigcirc (\neg f/a)$ (witness: $\mathit{afn}$).
	
	%as shown by this model:

%\begin{center}
%\begin{minipage}{0.5\textwidth}
%\begin{tabular}{ll}
%ideal: & $\neg a, \neg f$ $\quad $ $ a, f$ \\
%non-ideal: & $a, \neg f$ $\quad $ $ \neg a, f$  %\end{tabular}%\captionof{table}{VC}

%\end{minipage}
%\end{center}

\end{example}
Complementing Prop.\;\ref{inc:cor}, Prop.\,\ref{no-drow} gives sufficient conditions under which an overridden obligation does not drown. 
%For simplicity,  Prop.\,\ref{no-drow} considers a single overriding obligation; the result extends to multiple ones.
The condition $\Delta \mid\sim \neg (\top \Rightarrow (A \rightarrow B))$ indicates that the obligation in question is not expected to be necessarily fulfilled in the most normal worlds—arguably a separate issue.
%The condition $\Delta \mid\sim \neg (\top \Rightarrow (A \rightarrow %B))$ just says that one cannot expect the obligation be fulfilled in the %most normal worlds--this is a distinct concern.
%The condition  $\Delta\mid\sim\neg (\top\Rightarrow (A\rightarrow B))%$ states that
%$\kbc$ does not allow to form the expectation that the obligation will %be fulfilled$-$this is a different matter.

%, compliance with the obligation is not expected and, in this regard, %independent of normality considerations.

% says that, relative to $\kbc$, not-violating the obligation is not typical.The %assumption seems innocuous. 

%The clause ``$\m(\kbc)\not\models_{\PL} A\rightarrow B$'' says that $A\wedge\neg %B$ is not exceptional for $\kbc$.%$-$which is not to say that, relative to $\kbc$, it is %typical. 

 %implies that  $\neg (\top\Rightarrow A \rightarrow B)$ holds in the model.
%Thus, 
%$A\wedge\neg B$ is not exceptional for $\kbc$.
 %relative to 
%$\kbc$, a non-violation is not the typical situation. 

%can be understood as saying that violating the obligation is not exceptional, given $\kbc$. 
%no other obligation overrides the obligation in question;   $A\wedge C$  is exceptional for $%\kbc$. neither $A\wedge\neg B$ nor  

%Prop. \ref{no-drow} gives sufficient conditions for an %overridden obligation not to ``drown".

\begin{proposition}[No-drowning]\label{no-drow} Let $\Delta$ be such that $\bigcirc (B/A)\in\kbo$, $\exists r\in\kbo$ s.t. $r\triangleright \dob (B/A)$,
%, no obligation in $\kbo$ other than  $\bigcirc (\neg B/A\wedge C)$ overrides   $\bigcirc (B/A)$, 
%$b(r)$ is PL-consistent, 
$\top\Rightarrow\neg b(r) \in\kbc$ for all  $r\in\kbo$ s.t. $r\triangleright \dob (B/A)$,  $\Diamond (A \wedge B) \in\Gamma$ and
$\Delta \mid\sim\neg (\top\Rightarrow (A\rightarrow B))$. 
%$\m(\kbc)\not\models_{\PL} A\rightarrow B$.
% $A\wedge \neg B\not\in\epsilon(\kbc)$. 
%$A\wedge \neg B$ is not exceptional for $\kbc$. 
Then 
	\begin{flalign*}
		\Delta\mid\sim \bigcirc (B/A)
	\end{flalign*}
\end{proposition}
\begin{proof}  %We show that, for all ${\bf R}$-ordered model $M$, and all world $w$ in which all of $\Gamma$ holds, 
%$w\models\dob (B/A)$.
 Suppose $w\models\bigwedge\Gamma$.  As before, since $w\models \Diamond (A \wedge B)$, there is  $s\in \max_{\succeq_N}(\Vert A\wedge B \Vert) $. Let $u\in \max_{\succeq_N}(\Vert A\wedge \neg B \Vert)$. Since $\Delta \mid\sim\neg (\top\Rightarrow (A\rightarrow B))$,  $w\models \neg(\top\Rightarrow (A\rightarrow B))$. Hence
 %  $\m(\kbc)\not\models_{\PL} A\rightarrow B$, $\bigwedge\m(\kbc)\wedge A\wedge %\neg B$ is PL-consistent. By repleteness, 
  there is $v$ in $M$ s.t.
  %$v\models \bigwedge \m(\kbc)$ and $v\models A\wedge \neg B$. Since $v\models %\bigwedge \m(\kbc)$, we have
  	$v\in \max_{\succeq_N}(W)$ and $v\models A\wedge \neg B$. Let $t\in W$ be such that $t\succeq_N u$. By totality, either $t\succeq_N v$ or $v\succeq_N t$. In both cases, $v\succeq_N t$, since $v\in \max_{\succeq_N}(W)$. By totality again,  either $u\succeq_N v$ or $v\succeq_N u$. In both cases, 
	$u\succeq_N v$, since  $v\models A\wedge \neg B$ and $u\in \max_{\succeq_N}(\Vert A\wedge \neg B \Vert)$. By transitivity, $u\succeq_N t$,  and so $u\in \max_{\succeq_N}(W)$. By Prop.\,\ref{th:inc}, $w\models\top\Rightarrow \neg b(r)$  for all  $r$ s.t. $r\triangleright \dob (B/A)$. So $u\not\models b(r)$ for all such $r$.  Hence
%	Since 
%	no obligation in $\kbo$ other than  $\bigcirc (\neg B/A\wedge C)$ overrides   $%\bigcirc (B/A)$, 
$\bigcirc (B/A)\in V(u)$. Since
	$s\models A\wedge B$,  
	$\bigcirc (B/A)\not\in V(s)$. It follows that $u\not\succeq_I s$. Hence $w\models\bigcirc (B/A)$. 
	%By totality of $\succeq_N}$
%\vspace{-0.3cm}
%\phantom\qedhere	
\end{proof}

%It is with the no-drowning postulate that our account departs from Delgrande’s %original one, making the overall approach more suitable for normative reasoning. 

\section{Embedding into Constrained I/O logic}

%This section shows that the deontic fragment of the logic can be embedded into %constrained I/O logic (``cIOL'' for short).

\subsection{Basic Definitions}

%Constrained I/O logic (``cIOL'')  is defined on top of so-called 
This section begins with a brief description of unconstrained I/O logic \cite{MakinsonT00}, defined on top of classical propositional logic (PL). For readability, lowercase letters $a$, $x$, ... are used to represent Boolean formulas. A normative code is a set $N$ of conditional norms, defined as pairs of Boolean formulas $(a, x)$, where $a$ is the body of the norm and $x$ the head of the norm. The pair $(a, x)$ can
be read as ``if $a$, then $x$ is obligatory."
% In the unconstrainted I/O logic of \cite{MakinsonT00}, t
The main semantic construct is of the form $x \in \mathit{out}(N, a)$. This is interpreted as follows: given an input $a$ (state of affairs), $x$ (obligation) is included in the output under norms $N$.
$out$ is defined in terms of detachment (or modus ponens).  This paper uses the reusable throughput operation $out_4^{+}$ as the underlying I/O operation. It is based on the I/O operation  $out_4$, which is the strongest I/O operation among those studied by \cite{MakinsonT00}. $out_4^{+}$ is obtained from $out_4$ by allowing the input to reappear as output and adopting $(a, a)$ as a law of the logic. This leads to a collapse into classical logic. Specifically, $x \in out_4^{+}(N,a)$ if and only if $x \in out_4(N\cup \{(b,b) \mid b \text{ is a formula}\},a)$, which in turn is equivalent to $\{a\}\cup \m(N) \modpl x$, where $\m(N)$ is the materialization of $N$, with the definition straightforwardly adapted to the I/O setting.\footnote{Given $X\subseteq N$, $\m(X)$ is the set of all $b\rightarrow y$ with $(b,y)\in X$. When $X$ is a singleton set, curly brackets are omitted.}

%that  is, the set of all formulae $b\rightarrow y$ with $(b,y)\in N$.

To deal with contrary-to-duties (CTDs) and defeasible reasoning, constrained I/O logic (cIOL) adds a nonmonotonic mechanism that filters excess output via a set $C$ of formulas, called constraints.\footnote{See \cite{DBLP:journals/jphil/MakinsonT01,Parent11}.}  
To avoid excess output—namely, inconsistency—when too many norms apply, the set of norms is cut back just below the point of excess; this is the ``threshold'' idea,\footnote{For a comparison with Horty's approach, see \cite{Parent11}.} captured by the notions of maxfamily and outfamily (here, $out$ is $out_4^{+}$):

%The strategy for eliminating excess output is to cut back the set of norms to just %below the threshold of yielding excess—the resulting output is considered. This is the %``threshold'' idea,
%\footnote{For a comparison with Horty's approach, see \cite{Parent11}.} 
%captured by the notions of maxfamily and outfamily ):

%One ``cuts back'' the set of norms to
%just below the threshold of yielding excess, and consider the
%resulting output.  This is the ``threshold" idea.\footnote{For a comparison with %Horty's approach, see \cite{Parent11}.} It is captured by the notions of maxfamily %and outfamily (here, $out$ will be $out_4^{+}$):%Two more definitions are needed: 

\begin{itemize}
\item maxfamily: $\maxf(N,a,C)$ is the set of $\subseteq$-maximal subsets $H$ of $N$ such that $out (N,a)$ is PL-consistent with $C$;
\item outfamily: $\outf(N,a,C)=\{out(H,a)\mid H \in\maxf(N,a,C)\}$.
%\item $\cup\outfamily(G,A)=$
%\item
%$(a,x) \in \out_{\cup}(G)$ iff $x \in \cup\outfamily(G,A)$
%\\
%\indent\ \  $(a,x) \in \out_{\cap}(G)$ iff $x \in \cap\outfamily(G,a)$
\end{itemize} 
The so-called full meet constrained output under $N$ given input $a$, 
 $\cap(\outf(N,a,C))$, is the final output.  
 
Here, 
$C=\{a\}$. This choice is motivated by \cite{H69}’s view of the input as ``settled'': what has occurred cannot be undone, and the output must therefore be consistent with the input (\emph{ought} implies \emph{can}).

\subsection{Embedding}
Let $\kbc =\emptyset$ and  $\kbo =\{r_1, ..., r_n\}$.  The basic idea is to rewrite $\bigcirc (x/a)$ as an I/O pair $(a,x)$ with a more detailed body; this one also contains the negation of the bodies of all its defeaters, viz. all the more specific obligations overriding it. Let $\kbo_{\triangleright}$ represent the result of rewriting the obligations in $\kbo$ in this manner, and $\bigcirc_{\mathrm{H}} (x/a)$ stand for the Hanssonian obligation operator, defined by putting $\dob_{\mathrm{H}} (x/a)$ as true whenever $\max_{\succeq_I}(\Vert a\Vert) \subseteq \Vert x\Vert$.\footnote{Cf. \cite{H69,mak94}.} Let  $\kb =(\emptyset,\kbo)$.  The main result of this section is that  $(\emptyset,\kb)\mid\sim\bigcirc_{\mathrm{H}} (x/a)$ iff
 $x$ is in the full meet constrained output under $\kbo_{\triangleright}$ given input $a$, with 
 %the reusable throughout operation 
 $out_4^{+}$ as the underlying I/O operation and $C=\{a\}$. Thus, the embedding is faithful, viz.  it preserves and reflects validity.\footnote{
 Validity-preservation (sometimes called soundness) is the left-to-right direction of the bi-conditional in Th.\,\ref{iol:main}; validity-reflection is the right-to-left direction, for which no standard name exists.}
 % Validity-preservation is usually called soundness or correctness. The converse implication does not seem to have a standard name in the modal logic literature when considered on its own. For lack of a better name, I call it validity-reflection.}  
 Validity-reflection means that the translation does not output more obligations. From this, a partial result   (of validity-reflection)  for $\dob (-/-)$ under the $\exists\forall$ rule (cf. Prop.~\ref{thm:collapse})  follows. Even thought the result is partial, it gives some insights into the relationship between the two accounts.

The result and the main steps of the proof are given below; full details are provided in
Appendix~\ref{App2}. %The notation $\mathit{b(-)}$ and $\mathit{h(-)}$ is extended to I/O pairs in a straightforward manner. 
%The considered formulas are $\Box$-free and $\Rightarrow$-free. 
It is assumed that, on the preference side, the given model is replete.

%In what follows, I state the result, and describe the main steps of the proof. The details of the  proof are given in the appendix. The notation $\mathit{b(-)}$ and $\mathit{h(-)}$ is extended to the I/O pairs straightforwardly. I also assume that (on the preference side) a given model is always ``replete", viz., every consistent Boolean formula is satisfiable in it.

%\todo{define maxfamily}
%   as the underlying unconstrained operation.

To each $r_i:= \bigcirc (x/a)$, one associates the set of all its defeaters $D (r_i)$: $D (r_i)=\{r_j \in\kbo: r_j\triangleright r_i\}$. Define
\begin{flalign*}r_i^{\triangleright} =
\begin{cases}
 (a \wedge \bigwedge_{r_j\in D(r_i)}  \neg b(r_j), x) \text{ if } D (r_i)\not=\emptyset\\
 (a, x) \text{ otherwise }
\end{cases}
\end{flalign*}
Set $\kbo_{\triangleright}=\{ r_1^{\triangleright}, ..., r_n^{\triangleright}\}$.  %We have:
%For $H\subseteq\kb^{\triangleright}$,
%$m(H)$ is the set of all materializations of elements of $H$, viz., $\{a\rightarrow x: (a,x)\in H\}$. If $H=\emptyset$, then %$m(H)=\emptyset$ and
 %$\bigwedge m(H)$ is $\top$.

%\begin{restatable}{lemma}{littlelem} \label{little:lem1}
%$r_i\in V(w)$ iff $w\not\models m(r_i^{\triangleright})$. 
%\end{restatable}
%\begin{proof} See Appendix \ref{App2}.
%\end{proof}

%\begin{restatable}{lemma}{littlelemm} \label{little:lem2}
% If $w_1\succ_I w_2$, then $\exists r_i\in \kbo$ s.t. $w_1\models m(r_i^{\triangleright}) $ and 
%$w_2\not\models m(r_i^{\triangleright})$.
%\end{restatable}

%\begin{proof} See Appendix \ref{App2}.
%\end{proof}

%As mentioned the background input/output logic is so-called reusable throughout %$out_4^{+}$. Recall that $x\in out_4^{+}(\kb^{\triangleright},a) $ iff 
% $x\in out_4(\kb^{\triangleright}\cup\{(b,b): b \mbox{ is a formula}\},a) $.

%$a\cup m(\kb^{\triangleright})\models x$, where $\models$ is the consequence relation in %propositional logic. To obtain the reusable $out_4$, one needs not to allow the input be %resu-sed as an 

Th. \ref{ioltopref} and \ref{preftoiol} provide the central bridge: an 
$a$-world is best exactly when, for some member of the maxfamily, the world satisfies the conjunction of the elements of that member’s materialization.

 %it makes true each element in the materialization of a member of the maxfamily.
%Th. \ref{ioltopref} and \ref{preftoiol} are the key. They state that a best $A$-world is one that satisfies all the materializations of a member of the maxfamily.

\begin{restatable}[cIOL2Pref]{theorem}{bridgeRTL}
	\label{ioltopref}
	If $w\models a$ and $w\models \bigvee_{H\in\maxf (\kbo_{\triangleright},a,a)}\bigl(\bigwedge \m(H)\bigr)$, then 
	$w\in\max_{\succeq_I}(\Vert a\Vert)$, with $out_4^{+}$ as the background I/O operation. 
\end{restatable}
%\begin{proof} See Appendix  \ref{App2}.
%\end{proof}

\begin{restatable}[Pref2cIOL]{theorem}{bridgeLTR}\label{preftoiol}  
	If $w\in\max_{\succeq_I}(\Vert a\Vert)$, then both $w\models a$ and $w\models \bigvee_{H\in\maxf(\kbo_{\triangleright},a,a)}\bigl(\bigwedge \m(H)\bigr)$, with  $out_4^{+}$ as the background I/O operation. 
\end{restatable}

The following follows:

 \begin{restatable}[Faithfulness, Hanssonian truth-conditions]{theorem}{bridgecoro}
	\label{iol:main} $(\emptyset,\kb)\mid\sim\bigcirc_{\mathrm{H}} (x/a)$ iff $x\in \cap (\outf (\kbo_{\triangleright},a,a))$, with  $out_4^{+}$ as the background I/O operation. 
\end{restatable}

\begin{restatable}[$\exists\forall$ truth-conditions]{theorem}{bridgecoroEA}
	\label{iol:main2} $(\{\Diamond a\},\kb)\mid\sim\bigcirc (x/a)$ if $x\in \cap (\outf (\kbo_{\triangleright},a,a))$, with  $out_4^{+}$ as the background I/O operation. 
\end{restatable}

This embedding implies that the entailment problem for $\bigcirc_{\mathrm{H}}$ is in $\Pi^{p}_2$ (utilizing the corresponding result for I/O logic \cite{Sun17}). This matches the complexity of similar nonmonotonic formalisms, like default logic under the skeptical approach \cite{Got92}.
%while the monotonic basis
% (\AA qvist’s system {\bf E}) lies in co-NP \cite{CiabattoniOP22}.
% Second, it suggests a proof-theoretic characterization within the framework of adaptive logic \cite{StrasserBP16}.  

\section{Related and Future Work}

The combination of normality with deontic modalities has been considered before, but mostly within a monotonic setting (e.g., \cite{ddl:C80}). The strengthening problem can be related to what Lehmann calls the “presumption of typicality,” alluded to in footnote \ref{myfootnote}: properties of a class are assumed to hold for all its members unless there is evidence to the contrary. To quote Lehmann, “the presumption of typicality begins where rational monotony leaves off,” and its modeling therefore requires going beyond the standard KLM systems. One such approach is Lehmann's lexicographic entailment or the system proposed in \cite{Delgrande20}. Other non-monotonic candidates include the multi-preference closure developed in a description logic setting in, e.g., \cite{0001D20}. To my knowledge, the strengthening problem has not been explicitly addressed in deontic logic within a preference-based framework. 
%A possible case is \cite{bart98}, but the fallacy of the prohibited exception applies.

The present framework provides the basis for a number of substantial extensions. Chief among them is the design of a conflict resolution procedure based on a priority relation among obligations, intended to capture varying degrees of strength \cite{Horty07,Hansen06,serg25}. 
A key problem is to relate this priority relation to the ideality ordering on worlds in such a way as to yield prioritized counterparts of the postulates mentioned above, or suitable variants thereof \cite{ddl:PS97,Delgrande20}. 
%The most straightforward approach is to distinguish two notions of %defeat: one based on specificity and one based on strength.
%This work lays the groundwork for a range of substantial extensions. Chief among them is the design of a conflict resolution procedure %based on a priority relation among obligations, intended to capture varying degrees of strength \cite{Horty07,Hansen06,serg25}.  A key %problem is to relate this priority relation to the ideality ordering on worlds in a way that makes it amenable to identify prioritied %counterparts of the postulates mentioned above. 
%\cite{ddl:PS97,Delgrande20}.  The easiest is to distinguish two notions of defeat, specificity-based vs. strength-based. 
Another promising direction is 
the  integration of permission-as-exception \cite{Buy86}. A further direction is the generalization of the embedding$-$not only to prioritized I/O logic \cite{Parent11}, but also to other norm-based formalisms.   
Finally, the mechanization of reasoning tasks in Isabelle/HOL, building on the automation of the monotonic basis of the deontic fragment of the logic \cite{BenzmullerPT20,ParentB24}, remains an open problem.

%the mechanization of reasoning tasks in Isabelle/HOL, building on the successful automation of the monotonic basis of the logic \cite{BenzmullerPT20,ParentB24}, remains to be explored.

%Automation of reasoning tasks in Isabelle/HOL following the work of  \cite{BenzmullerPT20, ParentB24}, confined to the monotonic level) remains to be explored.
%The relationship with \cite{ddl:PS97}'s logic of explicit obligation warrants further study. 

%Parent11,

%\begin{quote}
	\section*{Acknowledgements}Research funded by the Austrian
	Science Fund (FWF) under the LuCi project [10.55776/PAT1007025].
%\end{quote}

%\begin{quote}
%	\section*{AI Declaration}
%	The author has not employed any %Generative AI tools.
%\end{quote}

%This work lends itself to a number of developments, the
%most pressing of which is extending it 
%%we would like to extend the account 
%to support conflict resolution with the use of a priority relation among obligations capturing a relation of strengths among duties \cite{Horty07,Parent11,Delgrande20}. The main difficulty lies in understanding how the priority relation on obligations affects the ideality ranking of worlds in the model.  
%%Different approaches are possible, see e.g.  \cite{Delgrande20} for some of them.  
%Another important direction is the extension with permission and extending the embedding not only to prioritized I/O logic  \cite{Parent11}, but also other norm based farmeworks like deontic default logic.

%, and permission. 
%This work lends itself to a number of developments, the
%most pressing of which is th

\appendix

\section{Proof of Th. \ref{iol:main} and  \ref{iol:main2}} \label{App2}

%Two lemmas are needed. Their proof is in the file ``Supplementary material''.
\begin{restatable}{lemma}{littlelem} \label{little:lem1}
	$r_i\in V(w)$ iff $w\not\models \m(r_i^{\triangleright})$. 
\end{restatable}
%\begin{proof} See Appendix \ref{App2}.
%\end{proof}

\begin{proof}[Proof of Lem. \ref{little:lem1}] 
	Let $r_i$ be of the form $\bigcirc (x/a)$. 
	
	\noindent
	Left-to-right. Assume $r_i\in V(w)$. By Def.\,\ref{violation:o}, this  means that
	%	\begin{center}
		%	\begin{minipage}{6cm}
			\begin{description}
				\item [({\rm i)}] $w\models a\wedge\neg x$, and
				%\item [(ii)]$w_2\not\models x$, and
				\item [{\rm(ii)}] $w\not\models b(r_j)$ for all $r_j\in D(r_i)$
			\end{description}
			%		\end{minipage}
		%		\end{center}
	The verification is split into two cases:
	%\begin{center}
	%	\begin{minipage}{7cm}
		\begin{description}
			\item [Case 1:]  $D(r_i)=\emptyset$. In this case, $r_i^{\triangleright}$ is just $(a,x)$, and (ii) holds vacuously. By (i),  $w\not\models \m(r_i^{\triangleright})$. 
			\item [Case 2:] $D(r_i)\not=\emptyset$. In this case, $r_i^{\triangleright}$ is  $(a \wedge \bigwedge_{r_j\in D(r_i)}  \neg b(r_j), x)$.  By (i)-(ii), $w\not\models \m(r_i^{\triangleright})$ and so the claim is proved too.
		\end{description}
		%		\end{minipage}
	%		\end{center}
Right-to-left.  Assume  $r_i\not\in V(w)$. By Def.\,\ref{violation:o}, this  means that
%		\begin{center}
	%		\begin{minipage}{6cm}
		
		\begin{description}
			\item [{\rm (i)}] $w\not\models a$, or
			\item [{\rm(ii)}] $w\models x$, or
			\item [{\rm(iii)}] $w\models a\wedge\neg x$ but  $w\models b(r_j)$ for some $r_j\in D(r_i)$
		\end{description}
		%	\end{minipage}
	%	\end{center}
%$w\not\models m(r_i^{\triangleright})$. 
%while the latter means that
The argument is split into two cases:
%	\begin{center}
	%	\begin{minipage}{7cm}
		\begin{description}
			\item [Case 1:]  $D(r_i)=\emptyset$. In this case, $r_i^{\triangleright}$ is just $(a,x)$, and either (i) or (ii) applies. In both cases, $w\models \m(r_i^{\triangleright})$.
			\item [Case 2:] $D(r_i)\not=\emptyset$. In this case, $r_i^{\triangleright}$ is  $(a \wedge \bigwedge_{r_j\in D(r_i)}  \neg b(r_j), x)$.  By  (i)-(iii) $w\models \m(r_i^{\triangleright})$, so the claim is proved too.
		\end{description}
		%	\end{minipage}
	%		\end{center}
	\end{proof}
	
\begin{restatable}{lemma}{littlelemm} \label{little:lem2}
	If $w\succ_I v$, then $\exists r_i\in \kbo$ s.t. $w\models \m(r_i^{\triangleright}) $ and 
	$v\not\models \m(r_i^{\triangleright})$.
\end{restatable}

\begin{proof}[Proof of Lem. \ref{little:lem2}]  
	Suppose $w\succ_I v$. This implies $v\not\succeq_I w$. Hence, $V(v)\not\subseteq V(w)$. So $\exists r_i\in \kbo$ s.t.  $r_i\in V(v)$ and $r_i\not\in V(w)$. Lem. \ref{little:lem1} then establishes the claim.
	%On the other hand, $r\not\in V_{o}(w_1, \kb)$ implies
	%\begin{itemize}
	%\item  [(i)]  $w_1\not\models a$, or
	%\item   [(ii)]  $w_1 \models b$, or
	%\item   [(iii)]  $w_1\models b(r_j)$ for some $r_j\in D(r_i)$.
	%\end{itemize}
	%A similar case analysis yields that $w_1\models m(r_i^{\triangleright}) $. 
\end{proof}

\bridgeRTL*

\begin{proof}%[Proof of Prop. \ref{ioltopref}]

Let $w\models a$ and $w\models \bigvee_{H\in\mathit{maxf} (\kbo_{\triangleright},a,a)}\bigl(\bigwedge \m(H)\bigr)$. So 
 $w\models \bigwedge \m(H)$ for some $H\in\mathit{maxf}(\kbo_{\triangleright},a,a)$.
% [The case where $H=\emptyset$ seems to be covered by the argument below]

Assume
$w\not\in\max_{\succeq_I}(\Vert a\Vert)$. Hence $\exists v\models a$ s.t. $v\succ_I w$. So $V (v)\subseteq V (w)$. By Lem. \ref{little:lem2}, 
$\exists r_i\in \kbo$ s.t. $v\models m(r_i^{\triangleright}) $ and 
$w\not\models \m(r_i^{\triangleright})$.  Hence, $\{a\}\cup \m (r_i^{\triangleright})\not\modpl\bot$, and so $H\not=\emptyset$. (If not, $H$ would not be $\subseteq$-maximal, since $\emptyset\subseteq \{r_i^{\triangleright}\}$.) 
% I break the argument into two cases:
%\begin{description}
%\item [Case 1:] 
%$H=\emptyset$.  In this case, for all $r_j\in\kbo$, $\{a\}\cup \m (r_j^{\triangleright})\modpl\bot$. If not, $H$ would not be $\subseteq$-maximal, since for such a rule $r_j^{\triangleright}$, $\{a\}\cup \m (r_j^{\triangleright})$ would be PL-consistent. 
%Suppose $w\not\in\max_{\succeq_I}(\Vert a\Vert)$. Hence there is $v\models a$ with $v\succ_I w$. By Lem. \ref{little:lem2}, 
%$\exists r_i\in \kbo$ s.t. $v\models \m(r_i^{\triangleright}) $ and 
%$w\not\models \m(r_i^{\triangleright})$. This contradicts the fact that $\{a\}\cup \m (r_i^{\triangleright})\modpl\bot$. So $w\in\max_{\succeq_I}(\Vert a\Vert)$ as required.
%\item [Case 2:]  
%$H\not=\emptyset$.  %By Rem. \ref{rem1} 
Suppose to reach a contradiction that  $v\not\models \bigwedge \m(H)$. There is  
$r_j^{\triangleright}\in H$ such that $v\not\models \m(r_j^{\triangleright})$ while $w\models \m(r_j^{\triangleright})$.
% The former  implies that 
%$v\models b(r_j)$, $v\models\neg h(r_j)$
 %and  $v\not\models b(r_k)$  for all $r_k \in\kbo$ s.t.
% $r_k\triangleright  r_j$. 
By Lem. \ref{little:lem1}, $r_j\in V(v)$ while  $r_j\not\in V(w)$. This contradicts the fact that $V(v)\subseteq V(w)$. Hence
 $v\models \bigwedge \m(H)$.  But $v\models a$ and $v\models \m(r_i^{\triangleright})$. So  $out_4^{+}(H\cup\{r_i^{\triangleright}\},a)$ is PL-consistent with $a$. Since $w\models \bigwedge \m(H)$ and $w\not\models \m(r_i^{\triangleright})$, 
$r_i^{\triangleright}\not\in H$. As a result, $H$ is not $\subseteq$-maximal. Contradiction. Hence $w\in\max_{\succeq_I}(\Vert a\Vert)$. 
%\end{description}
%\vspace{-0.1cm}\phantom\qedhere
\end{proof}

\bridgeLTR*

\begin{proof}%[Proof of Prop. \ref{preftoiol}]
%I show the contrapositive. 
Let $w\models a$ and $w\not\models \bigvee_{H\in\mathit{maxf} (\kbo_{\triangleright},a,a)}\bigl(\bigwedge \m(H)\bigr)$.
So 
 ($\star$) $w\not\models\bigwedge \m(H)$ for all $H\in\mathit{maxf}(\kbo_{\triangleright},a,a)$. 
%We have $\mathit{maxf}(\kb^{\triangleright},a,a)\not=\{\emptyset\}$. Otherwise  by Rem. \ref{little:lem2} ($\star$) would mean that %$w\not\models \top$$-$an impossibility. 
%\begin{description}
%\item [Case 1:]  
%$\mathit{maxf}(\kb^{\triangleright},a,a)=\emptyset$.  In this case, for all $r\in\kb$, $\{a\}\cup m %(r^{\triangleright})\models\bot$. Suppose $w\not\in\max_{\succeq_I}(\Vert a\Vert)$. Hence there is $v\models a$ %with $v\succ_I w$. By Lem \ref{l1}, 
%$\exists r_i\in \kbo$ s.t. $v\models m(r_i^{\triangleright}) $ and 
%$w\not\models m(r_i^{\triangleright})$. This is not possible, because  $\{a\}\cup m (r_i^{\triangleright})%%\models\bot$. So $w\not\in\max_{\succeq_I}(\Vert a\Vert)$ as required.
%\item [Case 2:]  
%We have $\mathit{maxf}(\kb^{\triangleright},a,a)\not =\emptyset$. 
Let  $\kb^{+}=\{r^{\triangleright}: 
w\models \m(r^{\triangleright})\}$. 
%I break the argument into two cases:
\begin{description}
\item [Case 1:]  $ \kb^{+}\not=\emptyset$. Since $w\models a$,  $\bigwedge\kb^{+}\wedge a$ is satisfiable, and hence PL-consistent. Hence  $\kb^{+}\subseteq H$ for some $H\in \mathit{maxf}(\kbo_{\triangleright},a,a)$.  By ($\star$),
$\kb^{+}\subset H$. By $H\in \mathit{maxf}(\kbo_{\triangleright},a,a)$, $a\wedge \m(H)$ is PL-consistent. By repleteness, $\exists v$ s.t. $v\models a\wedge \bigwedge \m(H)$. 

Let $r_i:=\bigcirc (y/b)\in V(v)$. By  Lem.\,\ref{little:lem1}, $v\not\models \m(r_i^{\triangleright})$. Suppose $w\models \m(r_i^{\triangleright})$. By definition, $r_i^{\triangleright}\in \kb^{+}\subseteq H$. Since $v\models \bigwedge \m(H)$, $v\models \m(r_i^{\triangleright})$, a contradiction. Hence $w\not\models \m(r_i^{\triangleright})$ and so $r_i\in V(w)$, by  Lem.\,\ref{little:lem1}.  Hence $V(v)\subseteq V(w)$, and so $v\succeq_I w$.

Since $\kb^{+}\subset H$, there  is $r_j^{\triangleright}$ s.t. $r_j^{\triangleright}\in H$  and $r_j^{\triangleright}\not\in \kb^{+}$. Since $v\models\bigwedge \m(H)$, $v\models \m(r_j^{\triangleright})$, and so by Lem.\;\ref{little:lem1}, $r_j\not\in V(v)$.
Since $r_j^{\triangleright}\not\in \kb^{+}$, $w\not\models \m(r_j^{\triangleright})$. So  again by Lem. \ref{little:lem1}, $r_j\in V(w)$. Hence $V(w)\not\subseteq V(v)$, viz., $w\not\succeq_I v$.

The two imply $v\succ_I w$, so that  $w\not\in\max_{\succeq_I}(\Vert a\Vert)$.

\item [Case 2:]   $ \kb^{+}=\emptyset$. By Lem.\,\ref{little:lem1}, $V(w)=\kbo$. Intuitively, $w$ violates all the obligations in $\kbo$. Without loss of generality, assume that $\exists H\not=\emptyset\in \mathit{maxf}(\kbo_{\triangleright},a,a)$.  For, if $\mathit{maxf}(\kbo_{\triangleright},a,a)=\{\emptyset\}$, then 
one would get $w\models \bigwedge \emptyset$ (\emph{aka} $\top$), contradicting ($\star$). By definition,  $a\wedge \bigwedge \m(H)$ is PL-consistent, and thus by repleteness $\exists v$ s.t. $v\models a\wedge \bigwedge \m(H)$. $V(v)\subseteq V(w)=\kbo$, so $v\succeq_I w$. But $H\not=\emptyset$, and so $\exists r_k^{\triangleright}\in H $ s.t. $v\models \m(r_k^{\triangleright})$. By Lem.\,\ref{little:lem1},   $ r_k\not\in V(v)$, and so $V(w)\not\subseteq V(v)$. Hence  $w\not\succeq_I v$, and so $v\succ_I w$, which suffices for $w\not\in\max_{\succeq_I}(\Vert a\Vert)$.
\end{description}
%\vspace{-0.5cm}\phantom\qedhere
\end{proof}

\bridgecoro*

\begin{proof}%[Proof of Prop. \ref{iol:main}]

Left-to-right. Let $(\{\top\},\kb)\mid\sim\bigcirc_{\mathrm{H}}(x/a)$. Hence, for all $\kb$-ordered model $M$ and $w$ in $M$, $w\models\bigcirc_{\mathrm{H}} (x/a)$. Let $H\in \mathit{maxf} (\kbo_{\triangleright},a,a)$. Suppose, for a reductio, that $\{a\} \cup \m(H)\cup\{\neg x\}$ is PL-consistent. 
By repleteness, $\exists v$ s.t. $v\models \bigwedge \m(H)$, $v\models a $ and $v\not\models x$. Clearly, $v\models \bigvee_{H\in \mathit{maxf}(\kbo_{\triangleright},a,a)} \bigr(\bigwedge \m(H)\big)$.
By Th.\,\ref{ioltopref},  $v\in\max_{\succeq_I}(\Vert a\Vert)$, and so $v\models x$, a contradiction. Hence $\{a\} \cup \m(H)\cup\{\neg x\}$ is not PL-consistent. In other words,  
$\{a\} \cup \m(H)\modpl x$, which suffices for $x\in out_4^+ (H,a)$. Hence $x\in \cap (\mathit{outf} (\kbo_{\triangleright},a,a))$.

Right-to-left. Assume $x\in \cap (\mathit{outf} (\kbo_{\triangleright},a,a))$. Let  $M$ be an $\kb$-ordered model, and $w$ in $M$. Let $v\in\max_{\succeq_I}(\Vert a\Vert)$. By Th.~\ref{preftoiol}, $v\models a$ and $v\models \bigvee_{H\in \mathit{maxf}(\kbo_{\triangleright},a,a)} \bigr(\bigwedge \m(H)\big)$. So $\exists H\in \mathit{maxf}(\kbo_{\triangleright},a,a)$ s.t. $v\models\bigwedge \m(H)$. 
%Suppose $H=\emptyset$. In this case, $\mathit{outf} (\kb^{\triangleright},a,a))=out_4^{+}(\emptyset,a)= Cn_PL(x)$, and so $v\models x$.   Suppose $H\not=\emptyset$.  In this case, 
But $x\in out_4^+ (H,a)$, viz. $\{a\}\cup \m(H)\modpl x$. So $v\models x$. This shows that  $w\models\bigcirc_{\mathrm{H}}  (x/a)$. Hence, $(\{\top\},\kb)\mid\sim\bigcirc_{\mathrm{H}}(x/a)$.
%\phantom\qedhere
\end{proof}

\bridgecoroEA*

\begin{proof} Assume $x\in \cap (\mathit{outf} (\kbo_{\triangleright},a,a))$.  Let  $M$ be an $\kb$-ordered model, and $w$ in $M$ be s.t. $w\models \Diamond a$. Hence, $\exists v$ s.t. $ v\models a$. By smoothness, $\exists u\in \max_{\succeq_I}(\Vert a\Vert)$.  The same argument as used for the right-to-left direction of Th.\,\ref{iol:main} yields $w\models \bigcirc_{\mathrm{H}}  (x/a)$. Hence, $u\models a\wedge x$.
%$w\models \bigcirc_{\mathrm{H}}  (x/a)$, so  $\max_{\succeq_I}(\Vert a\Vert)%\subseteq \Vert x\Vert$. Thus,  $u\models a\wedge x$.
Let $s$ be s.t. $s\succeq_I u$ and $s\models a$.  We claim that $s\in \max_{\succeq_I}(\Vert a\Vert) $. Let $t$ be s.t. $t\succeq_I s$ and $t\models a$.  By transitivity, 
$t\succeq_I u$. Hence $u\succeq_I t$,  since $u\in \max_{\succeq_I}(\Vert a\Vert)$. By transitivity again, $s\succeq_I t$. So $s\in \max_{\succeq_I}(\Vert a\Vert) $, and 
so $s\models x$.  By Prop.\,\ref{thm:collapse}, $w\models \dob (x/a)$. Hence, $(\{\diamond a\},\kb)\mid\sim\bigcirc (x/a)$.
%\vspace{-0.cm}\phantom\qedhere
\end{proof}

\author{% 
Xavier Parent
   % Author name
    \affiliations
    %Affiliation
    Institute of Logic and Computation, TU Wien
    \emails
    xavier.parent@tuwien.ac.at
    %email@example.com    % email
}

%\fi
% Multiple author syntax
%\author{%
%First Author$^1$\and
%Second Author$^2$\and
%Third Author$^{2,3}$\and
%Fourth Author$^4$ \\
%\affiliations
%$^1$First Affiliation\\
%$^2$Second Affiliation\\
%$^3$Third Affiliation\\
%$^4$Fourth Affiliation \\
%\emails
%\{first, second\}@example.com,
%third@other.example.com,
%fourth@example.com

%\loadtheorems{kr26-instructions.thm}

%\externaldocument{kr26-instructi%ons}

%\begin{document}

%\maketitle

%\section{Proofs}

%This document contains proofs of claims in %\cite{Parent26b}.
%The document will be made available on the author’s website upon acceptance.

%contains proofs of propositions and facts stated without proof in the submitted paper.

\section{Proof of Fact \ref{lem:decreasing} and Prop. \ref{smoothn}}

\setcounter{fact}{0}
%\setbox0=\vbox{\input{fact1.tex}} 
%\ref{fact:first}

%\myfact

	\begin{proof} By induction.%\vspace{-0.3cm}
	\begin{description}	
	\item [Base case.]
For $i = 0$, we show $\mathcal{E}_1 \subset \mathcal{E}_{0}$.
		Let $A\Rightarrow B\in \mathcal{E}_{1}$. By definition, $A\Rightarrow B\in \kbc$ and  $\kbc =\mathcal{E}_0$. So $\mathcal{E}_1 \subseteq \mathcal{E}_{0}$.
Suppose, to reach a contradiction, that	$\mathcal{E}_0 \subseteq \mathcal{E}_{1}$. Hence $\kbc\subseteq \varepsilon (\kbc)$. So $\m (\kbc)\models_{\PL} \bigwedge_{r\in\kbc} \neg b(r)$, and thus $\kbc$ is not coherent, contrary to Assumption \ref{assm:cons}.  Hence, $\mathcal{E}_0 \not\subseteq \mathcal{E}_{1}$ and so $\mathcal{E}_1 \subset \mathcal{E}_{0}$.

		%		$\mathcal{E}_1 = \varepsilon(\mathcal{E}_{0})$, and thus $A\Rightarrow %B\in \kbc = \mathcal{E}_0$ as required. 
%\vspace{-0.5cm}	

\item [Induction step.] Let $i \ge 0$ and suppose (inductive hypothesis, I.H.) that 
%\paragraph{\it Induction hypothesis.}
%We assume that, for some $i \ge 0$,
\[
 \mathcal{E}_{i+1} \subset  \mathcal{E}_{i}.
\]
We show 
\[
 \mathcal{E}_{i+2} \subset  \mathcal{E}_{i+1}.
\]
Let $A\Rightarrow B\in\mathcal{E}_{i+2}$. By definition, $A\Rightarrow B\in\kbo$ and 
		$\m(\mathcal{E}_{i+1})\models_{\PL}\neg A$. By the I.H., 
$\mathcal{E}_{i+1}\subseteq\mathcal{E}_i$. Hence $\m(\mathcal{E}_{i+1})\subseteq\m(\mathcal{E}_{i})$, and so $\m(\mathcal{E}_{i})\models_{\PL}\neg A$. It follows that $A\Rightarrow B\in\mathcal{E}_{i+1}$. Hence $\mathcal{E}_{i+2} \subseteq  \mathcal{E}_{i+1}$.

Suppose, to reach a contradiction, that $\mathcal{E}_{i+1} \subseteq  \mathcal{E}_{i+2}$. Thus, $\mathcal{E}_{i+1} \subseteq \varepsilon (\mathcal{E}_{i+1})$. So $\m (\mathcal{E}_{i+1})\models_{\PL} \bigwedge_{r\in\mathcal{E}_{i+1}} \neg b(r)$, and thus $\kbc$ is not coherent,  contrary to Assumption \ref{assm:cons}.  Hence, $\mathcal{E}_{i+1} \not\subseteq  \mathcal{E}_{i+2}$, and so $\mathcal{E}_{i+2} \subset  \mathcal{E}_{i+1}$.

\end{description}

		\end{proof}

\setcounter{proposition}{0}

\begin{proof} The proof proceeds by contradiction.
 Let $w_1\models A$ and $w_1\not\in
\max_{\succeq_N} (\Vert A\Vert)$. Suppose that $\forall v.$ $v\succ_N w_1 \rightarrowtriangle v\not\in
\max_{\succeq_N} (\Vert A\Vert)$. By the irreflexivity and transitivity of  $\succ_N$,   $w_1$ starts an infinite sequence $\sigma$ of increasingly more normal worlds (all distinct) $... \succ_N w_n \succ_N ... w_2\succ_N w_1$.\footnote{Cf. e.g. \cite[Prop. 4]{Parent14}.}
Fix any world $w_i$ in this sequence, and let its associated tuple be
\[
%t(w_i)=
\langle n^{w_i}_1,\ldots,n^{w_i}_m\rangle .
\]
Since worlds are ordered lexicographically by their associated tuples (i.e., comparing the first coordinate at which two tuples differ), for any two adjacent worlds $w_{k+1}\succ_N w_k$ in the sequence there exists a smallest index $j$ for which 
%least index $j$ such that
\[
n^{w_{k+1}}_j < n^{w_k}_j,
\]
and for all $i< j$,
\[
n^{w_{k+1}}_{i} = n^{w_j}_{i}.
\]
As we move upward along the sequence $\sigma$, the smallest $j$ at which tuples differ strictly decreases.
%otherwise, an increase or repetition of $j$ would violate the strict lexicographic ordering. 
Hence  $\sigma$ induces an infinite strictly decreasing sequence of natural numbers:
%an infinite strictly decreasing sequence of indices
\[
j_1 > j_2 > j_3 > \cdots .
\]
However, since tuples have finite length $m$, no such infinite sequence of natural numbers exists. Contradiction.

\end{proof}
\section{Proof of Prop. \ref{th:inc}}

%\ref{th:inc}

%\input{theorems.tex}

\setbox0=\vbox{
\begin{restatable}[Inclusion, normality]{proposition}{mythm}
	\label{th:inc}
	Consider a theory $\Delta$ such that $A\Rightarrow B\in\kbc$   and $A\wedge\neg B$ is PL-consistent. Then (under the assumption of repleteness)
	\begin{align*}
		&  \Delta\mid\sim A\Rightarrow B %\label{thm:inclusionin} 
	\end{align*}
\end{restatable}
%}
} 
% First statement appears %here too
\setcounter{proposition}{4} 
\mythm

%According to Theorem~\ref{cond:inclusion} from %the other document

%~\repthm{cond:inclusion}

%\repthm{cond:inclusion}
%\begin{proposition}[Inclusion, normality]\label{thm:inclusion} Consider a theory $\Delta$ such that $A\Rightarrow B\in\kbc$   and $A\wedge\neg B$ is PL-consistent. Then (under the assumption of repleteness)
%	\begin{align}
%		&  \Delta\mid\sim A\Rightarrow B \label{inc:o:cor} 
%	\end{align}
%	
%	% and $ A\Rightarrow B\in \kbc $. Then,  for all  $w$ \mbox{in }$M$, 
%	%\begin{flalign}
%	%	&  M,w\models A\Rightarrow %B \label{inc}\tag{Inc}
%	%\end{flalign}
%	
%\end{proposition}
\begin{proof} The proof proceeds via three lemmas.\footnote{I have drawn on \cite[Th. 2 and 3]{Lehmann95}.}
The rank $\rr$ of a Boolean $A$ is defined to be the smallest $i$ for which $A$ is not exceptional for $\mathcal{E}_{i}$, viz 
$\rr (A)= \min\{i | \mathrm{m}({\mathcal{E}_i}) \not\models_{\mathrm{PL}} \neg A\}$.  A
formula $A$ that is exceptional for all  $\mathcal{E}_{i}$  has no rank. An example of a formula with no rank is $\bot$. %and $\rr (A)=\infty$. 

%case where $A$ has no rank?

\setcounter{lemma}{0}
\renewcommand{\thelemma}{S\arabic{lemma}}
\begin{lemma} \label{inc:l1}  $\rr(A)< 
	\rr(A\wedge\neg B) \leq m$, where $m$ is the order of $\kbc$.
	% If $A\Rightarrow B\in\kbc$, then $\rr(A)< 
	%	\rr(A\wedge\neg B)$ or $A$ (and hence $A\wedge\neg B$) has no rank. %(What if $A\wedge\neg B$ has no rank?).
\end{lemma}

%What is $A$ has a rank, but not $A\wedge\neg B$?

\begin{proof}[Proof of Lem. \ref{inc:l1}] 
%By Fact \ref{A-lem:decreasing2}, 
By definition of a LM-sequence, 
%$\mathcal{E}_{m=\infty}=\emptyset$.  By definition of $\varepsilon$, 
$\mathcal{E}_{\infty=m}=\varepsilon (\mathcal{E}_{m-1})$. Suppose, to reach a contradiction, that $\m(\mathcal{E}_{m-1})\models_{\PL}\neg A$. Then, $A\Rightarrow B\in \mathcal{E}_{\infty}$ would follow, contradicting Fact \ref{lem:decreasing2}, stating that $\mathcal{E}_{\infty}=\emptyset$.  Hence, $\m(\mathcal{E}_{m-1})\not\models_{\PL}\neg A$, and so $\exists i$ s.t. $ \rr (A)=i \leq m-1$. 
%s.t. $\rr (A)=i$
 We have $\m(\mathcal{E}_{i})\not\models_{\PL} \neg A$ and $\m (\mathcal{E}_{j}) \models_{\PL} \neg A$ for all $j<i$. 
 
 Since $A\wedge\neg B$ is PL-consistent,  
$\m(\mathcal{E}_{m}) =\emptyset \not\models_{\PL} \neg (A \wedge \neg B)$. Hence,  $A\wedge\neg B$ is assigned a rank. 
Since $\neg A\models_{\PL} \neg (A \wedge \neg B)$,   $\m ({\mathcal{E}_{j}}) \models_{\PL} \neg (A\wedge\neg B)$ for all $j<i$.  Also $\m(\mathcal{E}_{i}) \models_{\PL} \neg (A \wedge \neg B)$, since 
$\m (\mathcal{E}_{i-1}) \models_{\PL} \neg A$ so that
	$A\Rightarrow B\in\varepsilon(\mathcal{E}_{i-1})=\mathcal{E}_{i}$ (and thus $A\rightarrow B\in\m(\mathcal{E}_{i})$).
Hence  $i<\rr (A\wedge\neg B)\leq m$.
	\end{proof}

\begin{lemma}\label{inc:f0} 
	\begin{flalign}
			&\forall i \in \{0, ..., m\}, \Delta_{i}\subseteq \mathcal{E}_{i} \tag{$\alpha$}\label{inc:f1}
			\\
		&	\forall i \in \{0, ..., m\},	\mathcal{E}_{i} = \bigcup_{j=i}^{m} \Delta_j  \tag{$\beta$} \label{inc:f2}	
	\end{flalign}
	\end{lemma}
		
(\ref{inc:f2})  says that $\mathcal{E}_{i}$ is the union of all $\Delta_j$ from $i$ to $m$.

\begin{proof}[Proof of Lem. \ref{inc:f0}]  %(\ref{inc:f1})(\ref{inc:f2})
(\ref{inc:f1}) follows from the definition: $\forall i \in \{0, ...,m-1\}$, $\Delta_{i}= \mathcal{E}_{i}-\mathcal{E}_{i+1}$ ; $\Delta_{m}=\mathcal{E}_{\infty}$.  The proof  of (\ref{inc:f2}) is by induction on the number of nested sets in the LM-sequence when read in reverse order.  By Fact \ref{lem:decreasing}, 
$ \mathcal{E}_{m}\subseteq ... \subseteq \mathcal{E}_{i} \subseteq ... \subseteq \mathcal{E}_{0}$. 
%$\mathcal{E}_{0}\supseteq ... \supseteq \mathcal{E}_{i}\supseteq ... \supseteq \mathcal{E}_{m=\infty}$.  

\begin{description}
	\item[Base case:] $m-1$ (the second-to-last set).  Since 	$\mathcal{E}_{m}=\mathcal{E}_{m+1}=\emptyset$, the equality is verified:
	\begin{flalign*}
		\mathcal{E}_{m-1}=
		 (\underbrace{\mathcal{E}_{m-1}-\emptyset}_{=\Delta_{m-1}}) \cup \underbrace{\emptyset}_{=\Delta_{m}}	  
	\end{flalign*}

	\item[Induction step: ] as an I.H., we assume (\ref{inc:f2}) holds for 
	$\mathcal{E}_{i+1}$, viz. 
	\begin{flalign*}
		\mathcal{E}_{i+1}=
		\bigcup_{j=i+1}^{m} \Delta_j		
	\end{flalign*}
	We argue (\ref{inc:f2}) holds for 	$\mathcal{E}_{i}$:
	\begin{flalign*}
		\mathcal{E}_{i}=&	(\mathcal{E}_{i}-	\mathcal{E}_{i+1})\cup	\mathcal{E}_{i+1} \qquad \mbox{ (set-theory)}\\
					=&	(\mathcal{E}_{i}-	\mathcal{E}_{i+1})\cup	\bigcup_{j=i+1}^{m} \Delta_j\qquad\mbox{ (I.H.)} \\
			=&	\Delta_{i}\cup	\bigcup_{j=i+1}^{m} \Delta_j\qquad\mbox{ (def of } 	\Delta_{i}) \\
			=&	\bigcup_{j=i}^{m} \Delta_j\qquad
	\end{flalign*}

\end{description}
\end{proof}

\begin{lemma} \label{inc:l2} 
	$\max_{\succeq_{N}} (\Vert A\Vert)\subseteq\Vert B\Vert$.
\end{lemma}
\begin{proof}[Proof of Lem. \ref{inc:l2}] %Let  $A\Rightarrow B\in \kbc$ and $\rr (A)< 
	%\rr (A\wedge\neg B)$. Suppose
	
%	\begin{comment}
	Let $w\in\max_{\succeq_{N}} (\Vert A\Vert)$.
	As before let $\rr (A)=i$. By definition, $\m(\mathcal{E}_{i})\not\models_{\mathrm{PL}} \neg A$. Hence, 
$\bigwedge\m(\mathcal{E}_{i})\wedge A$ is PL-consistent. By repleteness, there is a world $v$ in $M$ s.t. $v\models\bigwedge\m(\mathcal{E}_{i})\wedge A$. To help the reader, Table \ref{relevant} presents the tuple associated with each world.
%$\langle  | \Delta_{m} \cap F(v) |,  \langle  | \Delta_{m} \cap F(w) |,
\begin{center}
\begin{table}[h]\centering
\begin{tabular}{ll}
& Tuple \\
& {\scriptsize $\langle |\Delta_{m-1} \cap F(v) |, ..., |\Delta_{i} \cap F(v) |, ... ,|\Delta_{0} \cap F(v) |\rangle$}\\
& {\scriptsize $\langle  |\Delta_{m-1} \cap F(w) |, ..., |\Delta_{i} \cap F(w) |, ... ,|\Delta_{0} \cap F(w) |\rangle$ }
\end{tabular}\caption{Relevant worlds and their tuples. %The level of specificity decreases from left to right. 
$i $ is the rank of $A$. }\label{relevant}
\end{table}
\end{center}
By Fact \ref{lem:decreasing}, $v\models\bigwedge\m(\mathcal{E}_{j})$ for all $j>i$.  By Lem.\,\ref{inc:f0} (\ref{inc:f1}), 
	$v\models\bigwedge\m(\Delta_{j})$ for all $j\geq i$. Hence,
	\begin{flalign*}
		& \forall j\geq i: \Delta_{j} \cap F(v) =\emptyset  
	\end{flalign*} 
	so that
	\begin{flalign}
		& \forall j\geq i: | \Delta_{j} \cap F(v)|= 0  \label{ff1}\tag{\#}
	\end{flalign} 
%Either (i) $F(w)=\emptyset$ or (ii)  $F(w)\not=\emptyset$. In case (i), one immediately gets $w\models B$, since $w\models A$ and $A\Rightarrow B\in\kbc$. So suppose case (ii) applies. Either (ii.1) $F(v)=\emptyset$ or (ii.2)  $F(v)\not=\emptyset$. In case (ii.1) one immediately gets $v\succ_{N}w$, contradicting the assumption that $w\in\max_{\succeq_{N}} (\Vert A\Vert)$. Suppose case (ii.2) applies. 
Assume that 
	\begin{flalign}
		& \exists j\geq i: | \Delta_{j}\cap F(w)| \not= 0 \label{ff2}\tag{\#\#}
	\end{flalign} 
	(\ref{ff1}) and  (\ref{ff2}) imply  $v\succ_{N}w$, a  contradiction. Hence 
\begin{flalign*}
		& \forall j\geq i: | \Delta_{j}\cap F(w)| = 0 
	\end{flalign*} 	
So	
	\begin{flalign*}
		& \forall j\geq i: \Delta_{j}\cap F(w) =\emptyset 
	\end{flalign*} 	
	This means that
	$\forall j\geq i$ and  $\forall r\in\Delta_j$,  $r\not\in F(w)$, and hence $w\models \m(r)$. In other worlds, \begin{flalign*}
		w\models\bigwedge \m (\bigcup_{j=i}^{m} \Delta_j)	 
	\end{flalign*} 
 By Lem. \ref{inc:f0} (\ref{inc:f2}),
	\begin{flalign*}
		w\models \bigwedge\m(\mathcal{E}_{i})
	\end{flalign*} 	
%In particular, $w\models\wedge\m(\mathcal{E}_{i})$.
 By Lem.\,\ref{inc:l1}, $i< \rr (A\wedge\neg B)$, and so $\m(\mathcal{E}_{i})\models_{\PL}\neg (A\wedge\neg B)$.  Hence $w\models \neg (A\wedge\neg B)$, from which $w\models B$ follows since $w\models A$. Hence $\max_{\succeq_{N}} (\Vert A\Vert) \subseteq\Vert B\Vert$ as required.
\end{proof}
%This completes the proof of  Prop. \ref{th:inc}.
\end{proof}

\bibliographystyle{kr}
\bibliography{horty}

\end{document}